\documentclass[%
 reprint,
nofootinbib,
 amsmath,amssymb,
 aps,
]{revtex4-1}

\usepackage{graphicx}
\usepackage{dcolumn}
\usepackage{bm}
\usepackage{algorithm}
\usepackage{algorithmic}
\usepackage{float}
\usepackage{subfig}
\usepackage{amssymb,amsfonts,graphicx,amsmath,amsthm}
\usepackage{mathtools}
\usepackage{array}
\usepackage{enumitem}

\newcommand{\mc}{\mathcal}
\newcommand{\mbb}{\mathbb}

\newcommand{\ra}{\rangle}
\newcommand{\la}{\langle}

\newtheorem{theorem}{Theorem}
\newtheorem{lemma}[theorem]{Lemma}
\newtheorem{cor}[theorem]{Corollary}
\newtheorem{proposition}[theorem]{Proposition}

\begin{document}


\title{Quantum-Assisted Clustering Algorithms for NISQ-Era Devices}

\author{Samuel S. Mendelson}
 \email{Samuel.Mendelson@navy.mil}
\author{Robert W. Strand}%
 \email{Robert.Strand@navy.mil}
\author{Guy B. Oldaker IV}
 \email{Guy.B.Oldaker@navy.mil}
\author{Jacob M. Farinholt}
 \email{Jacob.Farinholt@navy.mil}
\affiliation{%
 Strategic \& Computing Systems Department\\
 Naval Surface Warfare Center, Dahlgren Division\\
 Dahlgren, VA
}%

\date{\today}%

\begin{abstract}
In the NISQ-era of quantum computing, we should not expect to see quantum devices that provide an exponential improvement in runtime for practical problems, due to the lack of error correction and small number of qubits available. Nevertheless, these devices should be able to provide other performance improvements, particularly when combined with existing classical machines.  In this article, we develop several hybrid quantum-classical clustering algorithms that can be employed as subroutines on small, NISQ-era devices. These new hybrid algorithms require a number of qubits that is at most logarithmic in the size of the data, provide performance improvement and/or runtime improvement over their classical counterparts, and do not require a black-box oracle. Consequently, we are able to provide a promising near-term application of NISQ-era devices.
\end{abstract}

\maketitle

\section{\label{Sec:Intro}Introduction}

The noisy, intermediate scale quantum (NISQ) computing regime, a term first coined by Preskill in Ref. \cite{Preskill_2018}, refers to the computing regime we are expected to reside in for the next decade or so. NISQ computing assumes $\sim 50 - 100$ qubits and a universal, albeit imperfect, set of gates. We should not expect to realize practical quantum algorithms in this regime that will have an exponential computational improvement over existing classical algorithms. Nevertheless, it may be possible to construct NISQ-era quantum (or hybrid quantum-classical) algorithms that have some other performance advantage. In particular, there has been significant interest recently in the development of NISQ-era algorithms for machine learning. While the list of publications in this area is rapidly growing \cite{DunjkoEtAl_2018,FarhiEtAl_2018,FarhiEtAl_2017,Kandala2017,Huggins_2019,HavlicekEtAl_2018,SchuldEtAl_2018,CrossEtAl_2015,WilsonEtAl_2018}, well-known examples include the Quantum Approximate Optimization Algorithm (QAOA) \cite{FarhiEtAl_2014,FarhiEtAl_2014-2,FarhiHarrow_2016} and the Variational Quantum Eigensolver (VQE) \cite{Peruzzo_2014,McClean_2016}. For a more extensive review of research in the area of quantum machine learning, we refer the reader to \cite{BiamonteEtAl_2017}.

Quantum algorithms that demonstrate an exponential advantage over classical algorithms often assume fault-tolerance. Because threshold theorems that admit fault-tolerant constructions only exist for circuit-based models, these models are generally pursued for scalable systems. However, because fault-tolerance cannot be assumed in the development of NISQ-era quantum algorithms, it is reasonable to treat adiabatic NISQ algorithms on equal footing as their circuit-based counterparts.

A particular subset of adiabatic quantum computing that has received much attention is one in which the problem Hamiltonian is {\textit{stoquastic}}, that is, a Hamiltonian such that, when expressed in a predetermined, fixed basis, has the property that all of the off-diagonal elements are non-positive.  Generally, we assume that the ground state is then measured in this fixed basis, although some have considered more general scenarios in which the measurement basis can be adaptively varied \cite{Fujii_2018}.  If the measurement basis can be adaptively chosen, then Ref. \cite{Fujii_2018} has shown that stoquastic adiabatic quantum computing can efficiently simulate quantum computing, and is hence universal.  

If we utilize only a fixed measurement basis (hereinafter referred to as the ``computational basis''), then with few exceptions \cite{Hastings_2013,JarretEtAl_2016}, it is believed that this restricted class of computations can be efficiently simulated classically. Nevertheless, stoquastic adiabatic algorithms may be able to demonstrate a significant performance advantage over certain existing classical algorithms, especially in the realm of machine learning, where quality and robustness are often more important performance parameters than run-time.

In what follows, we will view stoquastic adiabatic quantum computing from a graph-theoretic perspective, in much the same way as \cite{Jarret_2018}. In this framework, we design hybrid quantum-classical clustering algorithms that require very few qubits and show a performance increase in quality, run-time improvement, or both over their classical counterparts.  Consequently, we believe that these algorithms will have great utility on NISQ-era adiabatic quantum computing devices.  We note that our approach is uniquely different from most previous adiabatic algorithms in that (1) the quantum subroutine is not the solution to an optimization problem, (2) the final Hamiltonian is not diagonal in the computational basis, and (3) we do not rely on a black box for our problem Hamiltonian, thereby making its integration with classical algorithms more straightforward.

\section{Definitions and Notation}

We begin with a graph $G=(V,E)$ with vertex set $V$ and edge set $E \subset \begin{pmatrix}V\\2\end{pmatrix}$.  In this paper, we consider undirected, weighted graphs with symmetric edge-weight functions.  Let $u\in V$ be a vertex in $G$ and $\omega:E\rightarrow \mathbb{R}^+$ be the edge-weight function.  We define the degree of $u$ as $d_u=\sum_{v\in V} \omega(u,v)$.\footnote{When $\omega(u,v)=1$ for all $(u,v)\in E$, $d_u$ corresponds to the standard definition of degree.}  We define the {\textit{graph Laplacian}} of $G$, $L(G)$, as 
\begin{equation}
L(G)_{uv}=
    \begin{cases}
		d_u & u=v\\
		-\omega_{uv} & u\neq v
	\end{cases}.
\end{equation}
Now let $W:V\rightarrow \mathbb{R}$ be a vertex-weight function.  We write $W$ as a diagonal matrix where $W_{uu}=W(u)$.  If $H$ is a stoquastic Hamiltonian, then we can decompose $H$ as $H=L(G)+W$ for some graph $G$.\footnote{When it is clear from context we omit the dependence of $L$ on $G$.}  If $S\subset V$, we define 
\begin{equation}
|\phi_S\rangle \coloneqq |S|^{-1/2}\sum_{v\in S} |v\rangle
\end{equation}
 to be the equal superposition of all states corresponding to the vertices in $S$.  We let $|\phi\rangle \coloneqq |\phi_V\rangle$.

The graph Laplacian has two important properties that we use throughout the rest of this paper.  If $\lambda_0\leq \lambda_1\leq \dots\leq \lambda_n$ are the eigenvalues of the graph Laplacian $L(G)$, then $\lambda_0=0$ and the algebraic multiplicity of $\lambda_0$ is equal to the number of connected components of $G$.  Second, if $\{G_i\}$ is the set of connected components of $G$, then $\{|\phi_{V_i}\rangle\}$ is an orthonormal basis for the eigenspace of $\lambda_0$, where $V_i$ is the vertex set of $G_i$.  As we consider systems in which the the ground state is degenerate, we denote the spectral gap of $H$ as $\gamma=\lambda_j-\lambda_0$, where $\lambda_j$ is the first eigenvalue strictly greater than $\lambda_0$.

We often consider functions on graphs with the {\textit{Dirichlet boundary condition}}.  If $S\subset V$ and $\delta S$ is the boundary of $S$, then $f:S\cup\delta S\rightarrow\mathbb{R}$ is said to have the Dirichlet boundary condition if $f(v)=0$ for all $v\in \delta S$.  We define Dirichlet eigenvalues of the Laplacian $L$ as 
\begin{equation}
\lambda_i^{(D)}=\min_{\substack{f\perp T_{i-1}\\f|_{\delta S}=0}} \frac{\langle Lf,f\rangle}{\langle f,f\rangle},
\end{equation}
where $T_{j}$ is the space spanned by the functions $f_k$ that attain $\lambda_k$ for $0\leq k\leq j$.  We call these $f_j$ the Dirichlet eigenvectors corresponding to the Dirichlet eigenvalues. For a more detailed overview of Dirichlet eigenvalues and spectral graph theory, we refer the reader to \cite{Chung_1997}.

\section{Grover's Algorithm \label{Sec:GenGrover}}

Adiabatic Grover's algorithm on an $N$-state system consists of an initial and final Hamiltonian and a schedule $s(t)$.  The initial Hamiltonian for Grover's is given by $H_I=\mathbb{I}-|\phi\ra\la\phi|$.\footnote{This is a standard initial Hamiltonian in many adiabatic algorithms.} Grover's final Hamiltonian is given by $H_{Grov}=\mathbb{I}-|m\rangle\langle m|$, where $|m\rangle$ is the state of interest and the ground state of $H_{Grov}$.  The quantum system is initialized as $|\phi\rangle$, the ground state of $H_I$, and the system is evolved according to $H(s)=(1-s)H_I+sH_{Grov}$ where the schedule $s$ is given by
\begin{equation}\label{Eq:GroverSchedule}
s(t)=\frac{1}{2\sqrt{N-1}}\tan\left(2\frac{\sqrt{N-1}}{N}\epsilon t-\arctan\left(\sqrt{N-1}\right)\right)+\frac{1}{2},
\end{equation}
with $s(t_I)=0$, $s(t_F)=1$, and $\epsilon$ determines a lower bound of the inner product between the final state of the evolution $|E_F\ra$ and $|m\rangle$ given by $\left|\langle E_F|m\rangle\right|^2 \geq 1-\epsilon^2$ \cite{AvronEtAl_2010}.  With this schedule, Grover's algorithm scales as $\mathcal{O}(\sqrt{N})$.

Grover's algorithm can be generalized to a system in which there are multiple states of interest.  Suppose in an $N$-state system we are interested in measuring any of the states in the set $M=\{|m_i\rangle\}_i$.  The initial Hamiltonian remains the same as $H_I$ while the final Hamiltonian becomes $H_{Grov}=\mathbb{I}-\frac{1}{|M|}\sum_i |m_i\rangle\langle m_i|$.  The schedule then changes to 
\begin{equation}\label{Eq:GroverRatio}
s(t)=\frac{1}{2\sqrt{r-1}}\tan\left(2\frac{\sqrt{r-1}}{r}\epsilon t-\arctan\left(\sqrt{r-1}\right)\right)+\frac{1}{2},
\end{equation}
where $r=N/|M|$.  With this schedule, the runtime of Grover's is purely a function of the ratio $r$. In particular, the runtime is independent of $N$.

Here, we give a graph-theoretic interpretation of Grover's algorithm on an $N$-state system.  Grover's final Hamiltonian can be decomposed as $H_{Grov}=L(G)+W$ where $G$ is the empty graph on $n$ vertices and $W_{u}=1-\delta_{uv_m}$ where $v_m$ corresponds with the state of interest $|m\rangle$.  We embed $G$ into a larger graph $H$ with one more vertex $v_{N}$.  The vertex $v_N$ is connected to every other vertex except $v_m$.  We can then recover Grover's system by imposing a Dirichlet boundary condition on $v_N$ and considering the Dirichlet eigenvalues and eigenvectors of $H$.  In practice, we can impose the Dirichlet boundary condition on $v_N$ by restricting $L(H)$ to the vertices in $G$.  We can see explicitly then, that this Dirichlet boundary condition returns $H_{Grov}$.\footnote{Alternatively, we can impose the boundary condition by introducing a very large weight on $v_N$.  This forces $v_N$ to be approximately orthogonal to the ground state of $L(H)$.  Intuitively, all vertices adjacent to $v_N$ are affected by this weight and are pushed out of the ground state as well.}  The submatrix obtained by this restriction is called the {\textit{reduced Laplacian}}.  We call the vertices on which we impose the Dirichlet boundary condition {\textit{marked vertices}}.  We use this method of marking vertices to give a generalization of adiabatic Grover's algorithm. See Figure \ref{Fig:GroverGraph} for the graph characterization of traditional Grover's algorithm.

\begin{figure}
\begin{centering}
\includegraphics[scale=0.4]{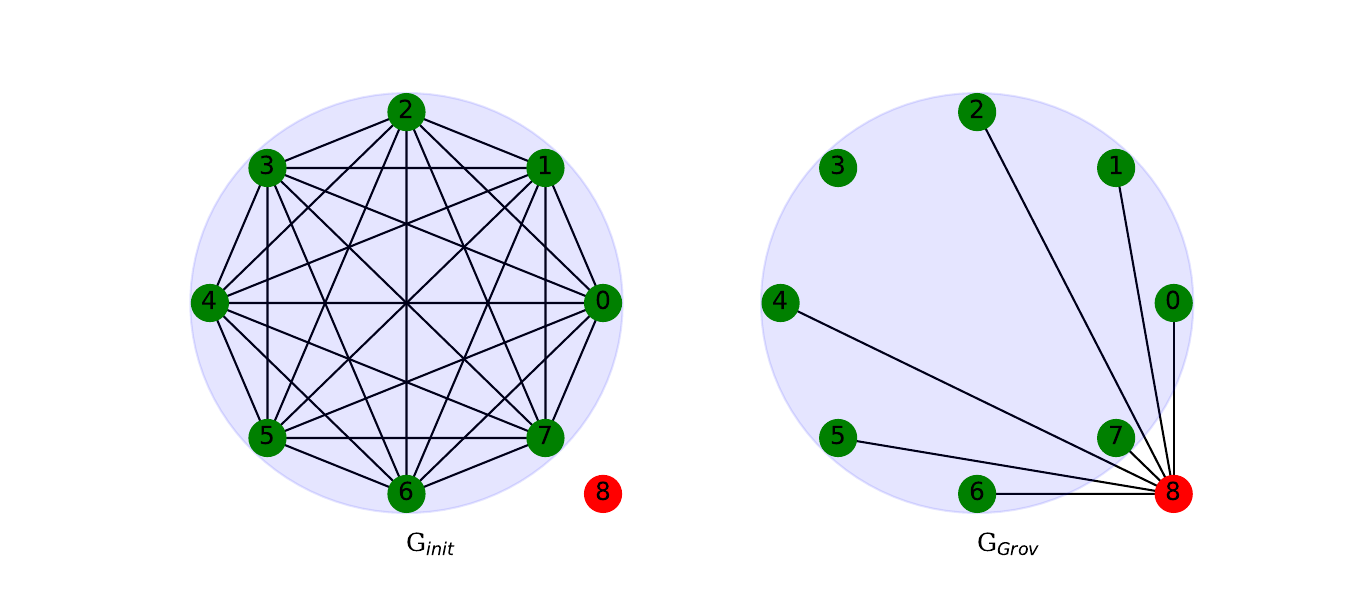}
\end{centering}
\caption{The graphs $G_{init}$ and $G_{Grov}$ corresponding to the initial and final Hamiltonians, respectively, in Grover's search algorithm over an eight element data set, where the desired element is the state $|3\ra$. The vertices in the shaded region correspond to elements of the unsorted database, while the outside red vertex acts as a ``dummy'' vertex. By deleting the row and column corresponding to the dummy vertex in both graph Laplacians, we obtain the initial and final Hamiltonians of adiabatic Grover's algorithm.}%
\label{Fig:GroverGraph}
\end{figure}

\subsection{\label{Subsec:GraphGroverEquiv}Generalization of Adiabatic Grover's}

Before we proceed, we state a simple but important structural lemma for the spectrum of a convex path through the space of Hermitian matrices. The proof is straightforward and is therefore omitted.

\begin{lemma}
Let $H:[0,1]\rightarrow \mbb{H}$ be a path through the vector space, $\mbb{H}$, of $(n \times n)$ Hermitian matrices, such that $H(s)=(1-s)H_0+sH_1$.  Let $A$ be the $m$-dimensional subspace of vectors that are eigenvectors of both $H_0$ and $H_1$ and $B=A^\perp$.
\begin{itemize}
\item[1)] There exists an orthornormal basis $\{f_i\}_{i=0}^{m-1}$ for $A$ such that $f_i$ is an eigenvector of $H(s)$ for all $s\in [0,1]$.
\item[2)] Let $\{\alpha_i\}_{i=0}^{n-m-1}$ be an eigenbasis of $H_0$ for $B$ and $\{\beta_i\}_{i=0}^{n-m-1}$ be an eigenbasis of $H_1$ for $B$.  The remaining eigenvectors and eigenvalues of $H(s)$ (those not in $A$) depend only on $\{\alpha_i\}$, $\{\beta_i\}$, and their corresponding eigenvalues.
\end{itemize}
\end{lemma}

This lemma tells us that the change in the spectral gap in the convex evolution of Hermitian matrices is uniquely determined by the eigenvectors in which the initial and final Hamiltonians differ. Applying this lemma to adiabatic Grover's algorithm gives us the following corollary.

\begin{cor}
Let $M\subset [N]$ and $B=[N]\backslash M$ where $M$ is a collection of states of interest.  The non-constant spectrum and corresponding eigenvectors of Grover's algorithm depend only on the eigenvectors and eigenvalues $\left(f_0=|\phi_M\rangle,\phi_0=0\right)$, $\left(f_1=|\phi_B\rangle,\phi_1=0\right)$, $\left(e_0=|\phi\rangle,\lambda_0=0\right)$, and $(e_1=\mc{N}(|B||\phi_M\rangle-|M||\phi_B\rangle),\lambda_0=1)$, where $\mc{N}$ is a normalization coefficient.
\end{cor}

\begin{proof}
The ground state of $H_I$, is $|\phi\rangle$ and has a ground energy of 0.  The ground state of $H_{Grov}$ is $|\phi_M\rangle$ and also has ground energy 0.  The first excited states of both are $(N-1)$-degenerate with a first-excited energy of 1.  The intersection of subspaces of the first-excited states is an $(N-2)$-dimensional subspace, and so $H_I$ and $H_{Grov}$ share an $(N-2)$-dimensional eigenspace.  The subspace on which they differ has an eigenbasis given by $e_0$ and $e_1$ for $H_I$ and an eigenbasis of $f_0$ and $f_1$ for $H_{Grov}$ as given above.  Therefore, by the above lemma, the spectrum and eigenspace of $H(s) = (1-s)H_I + sH_{Grov}$ is determined only by $(e_0,\lambda_0)$, $(e_1,\lambda_1)$, $(f_0,\phi_0)$, $(f_1,\phi_1)$.
\end{proof}

Using the above lemma and corollary, we give a graph-based generalization of adiabatic Grover's.  Suppose now we have a graph of $k$ connected components $\{G_i\}_{i=0}^{k-1}$ with $|V|=N$ and we are interested in finding a vertex in the $m$-th component, $G_m$ with $|V_m|=N_m$. We choose a set of vertices $V_d$ such that every vertex not in $G_m$ is adjacent to exactly one vertex in $V_d$.  We mark (impose the Dirichlet boundary condition on) the vertices in $V_d$ and use the corresponding reduced Laplacian as our final Hamiltonian $H_F$ (see Figure \ref{Fig:GenGroverGraph} for an example generalized Grover graph).  The initial Hamiltonian $H_I$ has all eigenvectors of $H_F$ except $f_0=|\phi_{G_m}\rangle$ and $f_1=|\phi_{\{G_i\}_{i\neq m}}\rangle$.  The corresponding eigenvalues are $\lambda_0=0$ and $\lambda_1=1$ respectively.  Therefore, while the spectra and eigenspaces of $H_F$ and $H_{Grov}$ are very different, by Lemma 1, the gap and ground state of this algorithm are equivalent to that of Grover's.

\begin{figure}%
	\subfloat[$G_{init}$]{
	\includegraphics[width=2in]{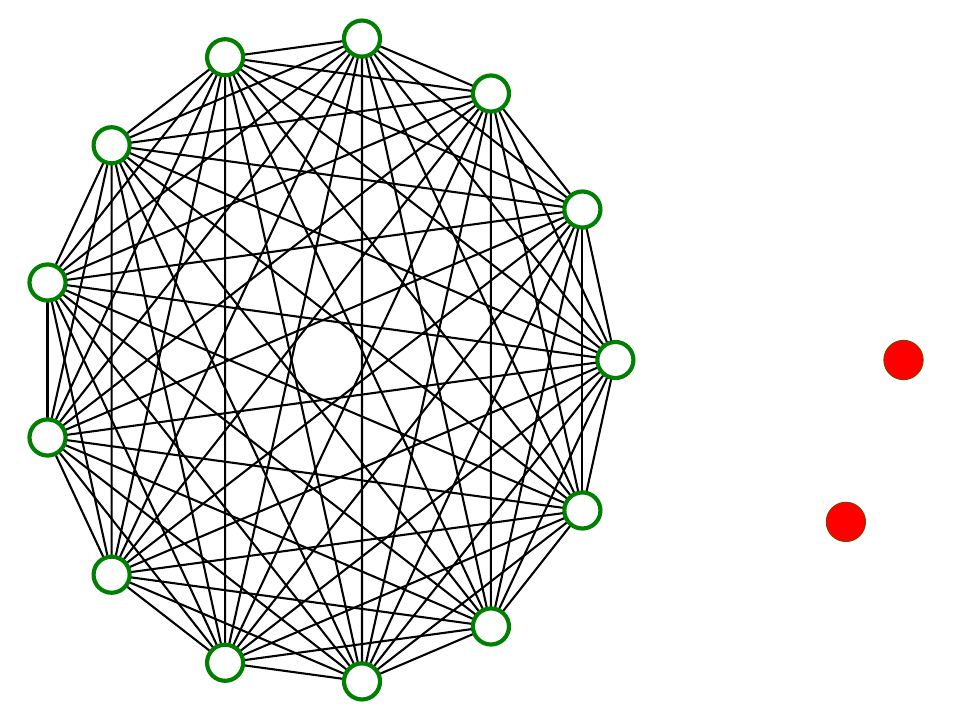}
	}

	\subfloat[$G_{Grov}$]{
	\includegraphics[width=2in]{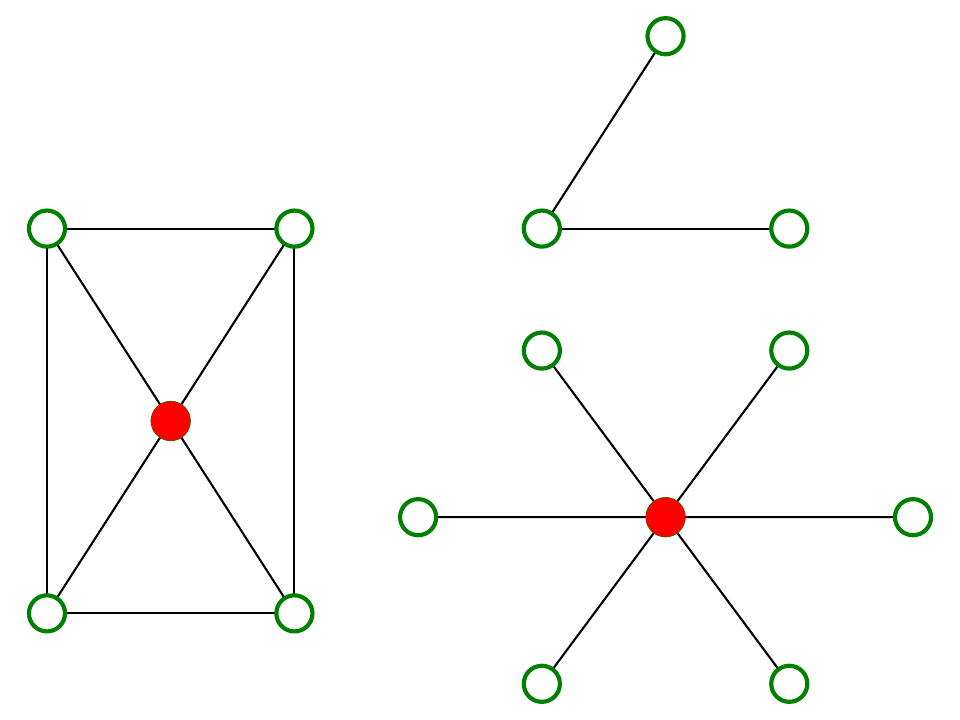}
	}
\caption{The graphs for the initial (a) and final (b) Hamiltonians of a generalized adiabatic Grover's algorithm. The red (filled in) vertices are the ``marked'' vertices that are removed to construct the reduced Laplacians that serve as the initial and final Hamiltonians. The ground state solution corresponds to the component in $G_{Grov}$ that has no marked vertex.}%
\label{Fig:GenGroverGraph}
\end{figure}

We see, then, that there are a large number of problem Hamiltonians that operate identically to the traditional adiabatic Grover's algorithm when the same initial Hamiltonian is used.

\subsection{Relaxation of Adiabatic Grover's Algorithm\label{Subsec:relax}}

While we have generalized the class of problem Hamiltonians for adiabatic Grover's algorithm, the task of finding the set $V_d$ and subsequently constructing $H_F$ is at least as difficult as constructing the traditional adiabatic Grover's oracle.  However, unlike Grover's algorithm, its generalization lends itself to relaxation.  We begin by relaxing the condition that every vertex in a connected component be adjacent to exactly one marked vertex.  While the evolution time $T$ for Grover's algorithm depends only on the ratio $N/|M|$, where $N$ is the total number of vertices and $|M|$ is the number of vertices that are not in a marked component, we do not expect this to be the case in general.  We begin with the worst case and consider the graph where half of the vertices are in a path, with an end vertex marked, and the other half are isolated (the configuration of the vertices in an unmarked component does not affect the evolution of the system, so without loss of generality, we may leave them isolated).  Since the path has a marked vertex, given a large enough evolution time $T$, the amplitudes of the path vertices tend to 0.  In Figure \ref{Fig:PathSims},  we plot the probability of measuring a vertex outside of the path against the total evolution time $T$ for both the optimal Grover's schedule of \eqref{Eq:GroverRatio} using $r=2$, and a constant schedule.  The value $N$ is the total number of vertices prior to marking. 
 We see that not only does the scaling depend on $N$, but that for a given $T$, the constant schedule is more appropriate than the optimal Grover's schedule.

\begin{figure}
    \centering
    \includegraphics[width=3.5in]{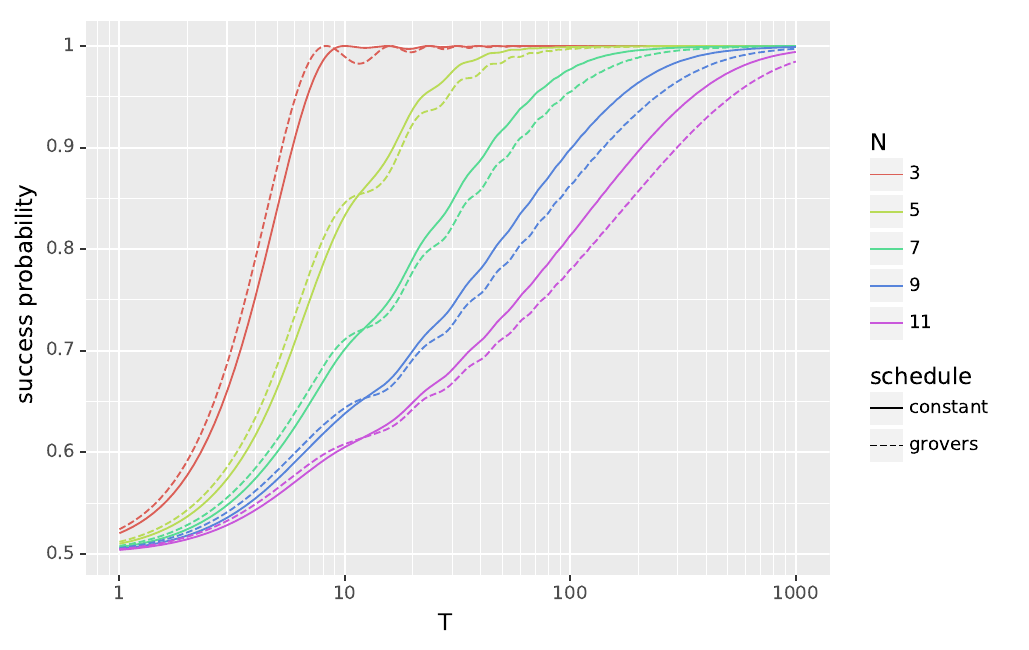}
    \caption{Performance of the Path example as we increase the total adiabatic evolution time $T$. We see that there is clearly a dependency on $N$ in this worst-case scenario.}
    \label{Fig:PathSims}
\end{figure}

We have shown, at least, that the generalized Grover's cannot be fully relaxed with regards to the connectivity of connected components. However, for certain applications, namely graph clustering, we do not expect to see such poor connectivity.  The task of graph-clustering involves partitioning a graph into pieces called \emph{clusters} with high connectivity such that each cluster has low connectivity with the other clusters.  We say that the intra-cluster connectivity should be high and the inter-cluster connectivity should be low. To this end, we look to find a representative of each cluster, and using this representative, partition the rest of the set. This scenario is a relaxation of our generalization in two ways.  First, for a given clustering problem, each vertex in a marked connected-component will not necessarily be adjacent to a marked vertex.  And second, more importantly, the components will not necessarily be disconnected from each other.  However, there is still useful information to be found in the ground state of $H_F$.

To better understand the ground state of $H_F$, it is helpful to view the ground state of $H_F$ in the context of random walks. The Dirichlet boundary condition is equivalent to imposing an absorbing state in the random walk on a graph $G$ \cite{JarretLackey_2017}. Given a distribution of walkers on $G$ with this absorbing state, the ground state of $H_F$ describes the quasi-stationary distribution of the graph walkers.  Walkers have a higher concentration on vertices from which it is difficult to reach the absorbing state.  When a graph is disconnected, this means the walkers concentrate on components not connected to the absorbing state.  When the graph is connected, walkers are concentrated on vertices not well connected to the absorbing state.  In the clustering case, vertices in clusters not containing a marked vertex have higher amplitudes than those that do contain a mark.  in Ref. \cite{JarretLackey_2017}, the authors use this quasi-stationary distribution to drive an optimization algorithm.  In our case, the quasi-stationary distribution is the end goal, as it provides information on cluster membership. This will be shown in more detail in Section \ref{Sec:QUAC}.

In the idealized (i.e. generalized Grover's) case, the inter-cluster connectivity is zero, and the intra-cluster connectivity is complete. If additionally each of the $k$ components is roughly the same size, then for large $N$, we may set $r = k$ in the schedule \eqref{Eq:GroverRatio}. As stated before, Grover's algorithm then has constant runtime for a fixed $k$, and is independent of $N$.  In the Appendix, we analyze a very weak relaxation of Grover's by adding an edge between two complete graphs on $N$ vertices with one marked vertex.  We see that in the limit as $N$ goes to infinity, the spectral gap converges to that of Grover's algorithm.

\begin{figure}
 \begin{centering}
  \includegraphics[width=0.4\textwidth]{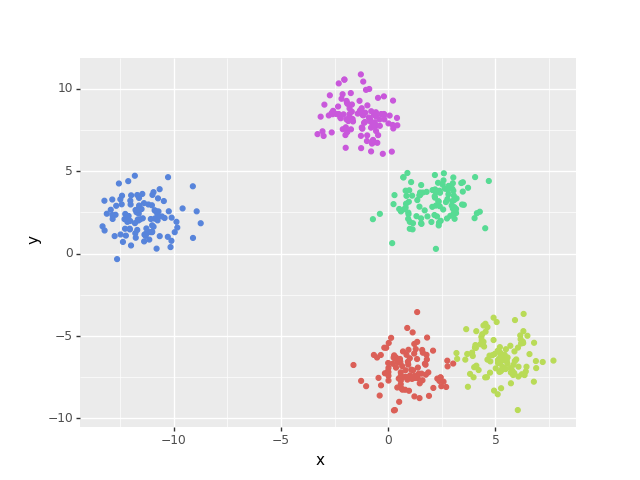}
 \end{centering}
 \caption{The Five Cluster dataset, correctly partitioned.}%
 \label{Fig:5Cluster}
\end{figure}

Figure \ref{Fig:5Cluster} provides an example of a dataset with five clusters. When this data is mapped to a graph (we discuss in the following section how we create graphs from vector-valued data), there is nonzero inter-connectivity between several of the clusters and provides an example of a weak relaxation of generalized Grover's. As demonstrated in Figure \ref{Fig:MarkSims}, we see empirically that there is only a weak dependence on $N$.  We expect this dependence to remain minimal for datasets that can be reasonably clustered.  As the structure of the data becomes less cluster-like (for example, in the connected path scenario), this size-independence property starts to break down.  In what follows, we develop a \emph{Quantum Cluster Initializer}, or \emph{QCI} algorithm, the goal of which is to find representatives from each cluster.

\begin{figure}
    \centering
    \includegraphics[width=3.5in]{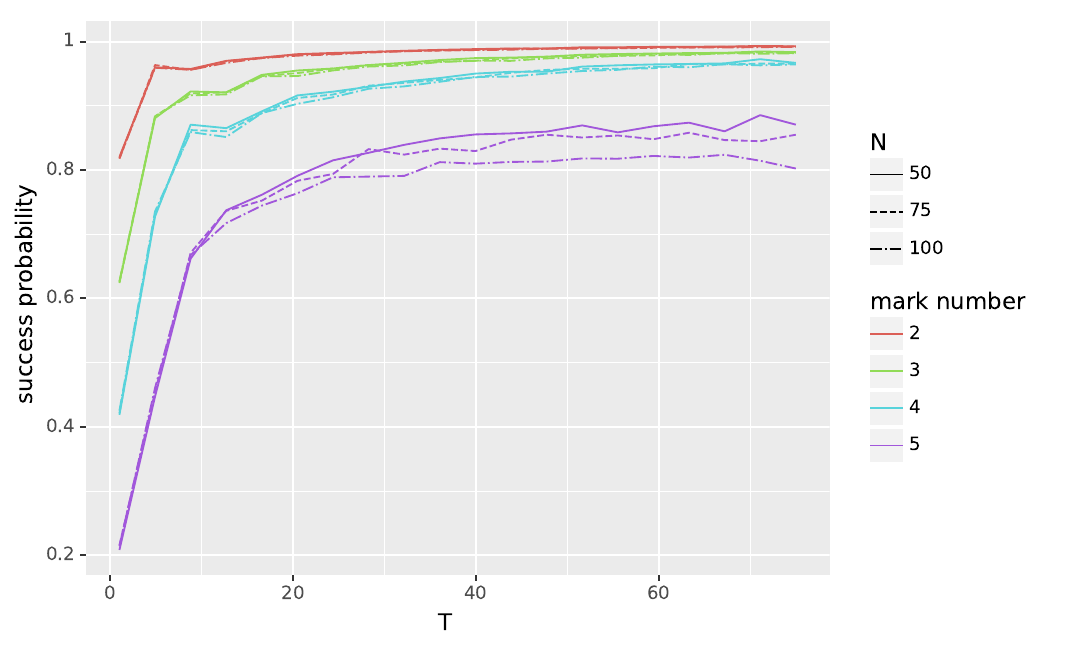}
    \caption{Performance of QCI on the Five Cluster dataset as a function of the overall runtime $T$. The colors correspond to the number of marked vertices. Uniform random subsets of the dataset of various sizes $N$ were chosen to run QCI over.}
    \label{Fig:MarkSims}
\end{figure}

The quantum subroutine itself is rather straightforward. Given a graph Laplacian and a set of vertices to ``mark,'' the quantum computer builds its problem Hamiltonian from the corresponding reduced Laplacian.\footnote{For the purposes of this algorithm, ``marking'' a vertex corresponds to deleting the corresponding row and column in the graph Laplacian.} We then adiabatically evolve our system accordingly, and measure the final ground state in the computational basis. If the graph has multiple components, then the ground state immediately prior to final measurement is a superposition of basis states representing vertices from unmarked connected components. Consequently, measuring the final state gives a representative from a different connected component, that is, a different cluster, with certainty. If, however, the graph is connected, but can be grouped into clusters with high intra-cluster connectivity and low inter-cluster connectivity, then the ground state of the problem Hamiltonian is a superposition of all states. The states with the largest amplitudes are those from unmarked clusters. Consequently, measuring the ground state has a high probability of collapsing it onto a representative of a previously unmarked cluster. 

\begin{figure}
 \centering
 \includegraphics[width=0.5\textwidth]{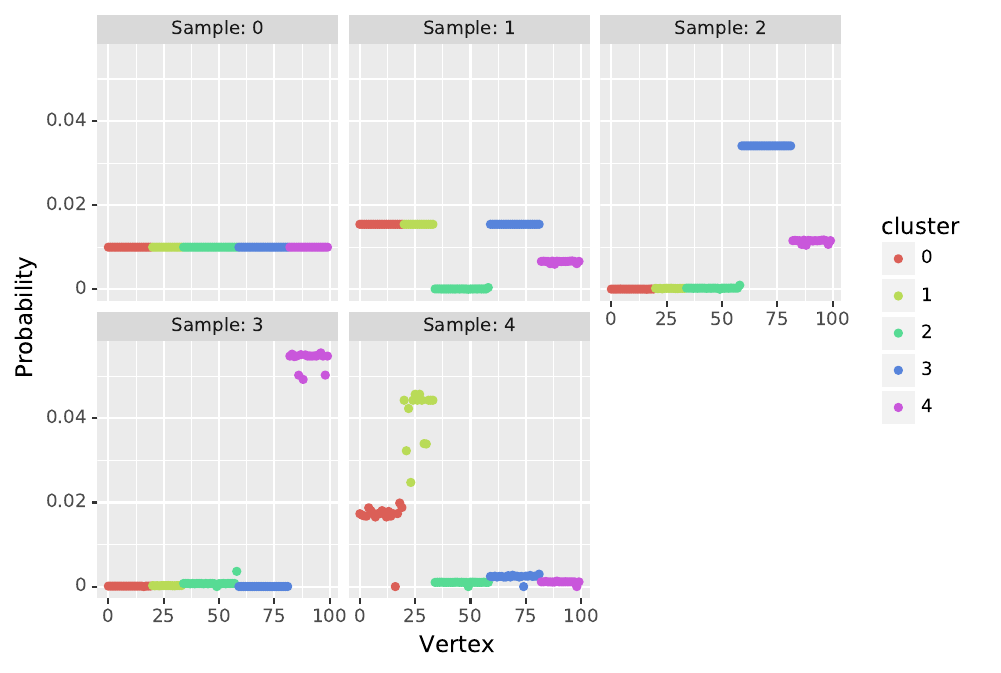}
 \caption{An example of the change in probability amplitudes associated with each vertex in the ground state solution as we run QCI to find (and consequently mark) a representative of each cluster in the Five Cluster dataset.}%
 \label{Fig:SampleProbs_5C}
\end{figure}

In Figure \ref{Fig:SampleProbs_5C} we provide an example of how the probability amplitudes of the vertices in the ground state solution change with each iteration of QCI run over a subset of a five cluster dataset from Figure \ref{Fig:5Cluster}. 
From Figures \ref{Fig:MarkSims} and \ref{Fig:SampleProbs_5C}, we see even with modest interconnectivity between some clusters, the QCI algorithm finds representatives of unmarked clusters.  We provide the pseudocode for QCI in Algorithm $\ref{alg:qci}$.


\begin{algorithm}[H]
    \caption{Samples vertex based on Dirichlet boundary conditions}
    \algsetup{indent=2em}
 \begin{algorithmic}[1]   

    \STATE \textbf{function QCI} $(L , marks)$\\
    \textbf{Input:} Laplacian $L$, set of marked vertices $marks$\\
    \textbf{Output:} $v_{meas}$
    \begin{ALC@g}
    \STATE $H_I\Leftarrow\mathbb{I}-|\phi\rangle\langle\phi|$\\
    \STATE $|\psi(0)\rangle\Leftarrow |\phi\rangle$\\
    \STATE $H_F \Leftarrow$ reduced Laplacian $L$ with marked vertices $marks$\\
    \STATE Evolve system\\
    \STATE $v_{meas}\Leftarrow$ Measurement
    \end{ALC@g}
    \RETURN $v_{meas}$

 \end{algorithmic}
 \label{alg:qci}
\end{algorithm}

In what follows, we assume that our data should be clustered;  that is, our data can be partitioned into highly intra-connected components with low inter-connectivity.

\section{QUAC\label{Sec:QUAC}}
In this section, we utilize our QCI algorithm in two scenarios.  First, as a seed initialization for an existing classical algorithm, $k$-means; and second, as the basis for a new unsupervised nearest-neighbor clustering algorithm.  We show that our $k$-means initialization performs just as well as the current standard, $k$++, on well-behaved datasets and outperforms it on others.  When compared to the current graph clustering standard, spectral clustering, our nearest-neighbor algorithm has comparable performance across the board, but with a polynomial improvement in runtime. Moreover, unlike spectral clustering, our algorithm does not require the graph to be connected. These quantum-assisted clustering algorithms, or \emph{QUACs}, are a hybrid between quantum computation and existing classical algorithms, and should serve as a realizable application of NISQ-era devices.

\subsection{Numerical Simulation Setup\label{Sec:NumSim}}

We consider the full performance of the hybrid algorithms and compare them against the purely classical counterparts. To simulate the QCI subroutine, we use an eighth-order Runge-Kutta method implemented in Python on the time-dependent Schr{\"o}dinger's equation:
\begin{equation}
i\hbar\frac{d}{dt}|\phi(t)\rangle= H(t)|\phi(t)\rangle,
\end{equation}  
where for simplicity, we set $\hbar = 1$. We follow Grover's schedule given in Equation \eqref{Eq:GroverRatio}, setting $r$ as the number of clusters $k$, based on the theoretical analysis and empirical evidence from the previous section.  The runtime $T$ is computed using an $\epsilon$ value of 0.01. We test our algorithms against several different test datasets we generated. Since QCI finds representatives of clusters in graphs, we must first convert our vector-valued datasets into graphs.  There are many methods for constructing graphs from vector-valued data; including, but not limited to, $\varepsilon$-ball methods, nearest-neighbor methods, and kernel methods.  There is a large collection of literature exploring which methods to build graphs from different kinds of data for different applications.  This discussion is outside the scope of this paper and so we refer interested readers to  Ref. \cite{DaitchEtAl_2009} and references therein for a more detailed introduction to this area.

In what follows, we use the $\varepsilon$-ball method to create our graphs. Each point in the dataset will be represented by a vertex. For a prescribed $\varepsilon$, two vertices will have an edge between them if their corresponding points have a Euclidean distance less than $\varepsilon$.  This type of graph construction does not create weighted edges. After the graph is constructed, outlier removal is performed, with any isolated vertices removed. 

Because we are assuming our quantum computer has only a small number of qubits and/or limited connectivity, we run the quantum subroutine on only a subset of the data (chosen uniformly at random), which we call the \emph{thinned} dataset.  In all cases, the size of the thinned set is $10\%$ of the full data.
After our quantum initialization is complete, we run the classical routine over the entire dataset, with the quantum results as the initialization.  The sampling rate for the evolved states is denoted $m$ at the beginning of each algorithm.  We discuss $m$ for each of our algorithms in their corresponding sections.

\subsection{$q$-means}
The $k$-means algorithm \cite{Lloyd_1982} is a standard clustering algorithm that is simple but effective.  Given a set of vector-valued data and a user-defined value $k$, the $k$-means algorithm partitions the data into $k$ clusters. To begin, a set of $k$ seeds is chosen at random from the data.  Next, the distance between each data point and each seed is calculated.  By associating each point with the closest seed, the data is partitioned into $k$ sets.  The centroid of each cluster is calculated as the average of the points in each cluster.  Using these centroids, the data is partitioned again with the centroids as the new seeds.  This process is repeated until either there is no update to the centroids or a maximum number of iterations is reached.

The goal of this algorithm is to minimize the inertia, that is, the sum of square distances to the centroids of each cluster. However, it is very sensitive to the initial choice of seeds and often finds a local minimum of inertia. In addition, if the initial seeds are poorly chosen, the algorithm may require a large number of iterations before convergence.

To combat these problems, a ``smarter'' seeding algorithm has been developed \cite{Arthur_2007}.  Known as $k$++, this seeding algorithm is based on the idea that cluster centers should be spaced far apart.  First, a point is chosen at random to be the first seed.  Next, the distance between the first seed and every other point is calculated.  Using the normalized distances as a probability distribution, the next point is chosen.  This process is repeated until all $k$ seeds are chosen.  The rest of the $k$-means algorithm remains unchanged.  While there is a higher computation cost up-front, it has been shown that this initialization reduces the number of iterations required and is an overall reduction in computation \cite{Arthur_2007}.  Empirically, the solutions tend to have better inertia as well.  Even with these improvements, $k$++ still produces poor results if the clusters are not spherical Gaussians with similar variation and density.  As we will see, the solution that minimizes the inertia is not always the ``correct'' solution.  We will show how to use our QCI subroutine to generate seeds for the $k$-means algorithm and show that our seeds give a performance on par with $k$++ for well-behaved data and a demonstrable performance improvement for a certain class of data.

We note that an adiabatic initialization algorithm for $k$-means was developed previously in Ref. \cite{LloydEtAl_2013}; however, the amplitudes of the ground state in that algorithm reproduce the probability amplitudes in the $k$++ algorithm, and consequently its ability to accurately cluster is equivalent to that of $k$++. Our ``$q$-means'' algorithm, on the other hand, uses QCI as a seed initialization for $k$-means.

We begin our initialization with a dataset and user-defined values $k$ and $m$ for the number of clusters and sampling rate respectively. We first build a graph and the corresponding Laplacian for the data as described in section \ref{Sec:NumSim}.  Once we construct the Laplacian, we run QCI $k$ times to obtain, ideally, a representative of each cluster.  While these points are representatives of the clusters, they are not necessarily good representatives of the cluster centers.  For a good approximation of the cluster center of the first cluster, we remove the mark associated with the first cluster, leaving the other marks in place, and run QCI $m$ times.  QCI returns different representatives of the first cluster after each run.  We average the data points associated with these vertices and use this for our initial centroid for the first cluster.  We repeat this process for each of the $k$ clusters.  The pseudocode for $q$-means is provided in Algorithm $\ref{alg:qmeans}$.


\begin{algorithm}[H]
    \caption{Returns seed initialization for $k$-means}
    \algsetup{indent=2em}
 \begin{algorithmic}[1]
    
    \STATE \textbf{function $q$-means} $(X,L, k, m)$\\
    \textbf{Input:} data set $X$, Laplacian of data $L$, number of clusters $k$, number of measurements for each cluster $m$\\
    \textbf{Output:} seed initializations $\{c_n\}_{n=1}^k$\\
    \begin{ALC@g}
    \STATE $v_1\Leftarrow$ vertex chosen uniformly at random\\
    \FOR{$i\in [k-1]$}
        \STATE $v_{i+1}\Leftarrow \mbox{QCI}\left(L,\{v_n\}_{n=1}^{i}\right)$
    \ENDFOR \newline\\
    \FOR{$i \in [k]$}
        \FOR{$j \in [m]$}
            \STATE $u_j \Leftarrow \mbox{QCI}\left(L,\{v_n\}_{n\neq i}\right)$
        \ENDFOR
        \STATE $c_i\Leftarrow \frac{1}{m}\sum_n x_{u_n}$
    \ENDFOR
    \end{ALC@g}
    \RETURN $\{c_n\}$
 \end{algorithmic}
 \label{alg:qmeans}
\end{algorithm}

To determine the complexity of $q$-means, we need to consider both the complexity of the quantum subroutine as well as the complexity of classical computation.  We call QCI $k$ times to find the initial cluster representatives, and then $k\times m$ times to compute the seeds.  With our choices of schedule, QCI is constant with respect to the problem size $N$.  Therefore, the quantum subroutine complexity depends only on $m$. For a given  probability distribution $P$ over a finite dataset $X$, the error of a finite approximation of the mean of $X$ over $P$ depends only on the standard deviation of $P$ over $X$ and the sampling rate $m$.  Therefore, assuming $X$ is bounded, the error on the approximation of the mean is independent of the cardinality of $X$.  Therefore, our $q$-means initialization is $\mathcal{O}(1)$ with respect to $N$.  The only computation cost we need to consider is the classical construction of the graph Laplacian from the data. The runtime of this construction depends on the method used to create the graph.  For a kernel-based method, the construction is $\mathcal{O}(N^2)$.  For a nearest-neighbor-based method, the construction is $\mathcal{O}(N\log^2 N)$. This is the limiting factor in the runtime of our $q$-means initialization, whereas $k$++ has runtime $\mathcal{O}(N)$. We show a performance improvement as well as a slight improvement in the number of iterations needed to converge for our synthetic data.

We simulate $q$-means on two different datasets.  The first is a set of points chosen from five circular Gaussian distributions, which we denote as the \emph{Five Cluster} dataset, first presented in Section \ref{Subsec:relax} (see Fig. \ref{Fig:5Cluster}).  The $k$++ initialization performs extremely well on this type of data.  The next dataset is a set of points chosen from two Gaussians; one of which is circular and one of which is elliptical (the variance in the $x$-direction is twice the variance in the $y$-direction) which we call the \emph{Elliptical} dataset (see Fig. \ref{Fig:Elliptical}).  This dataset is particularly difficult for $k$++ (and the random initialization) to cluster ``successfully.''  By success, we mean the algorithm partitioned the dataset according to those clusters we visually distinguish.  The inertia of the ``correct'' clustering of the elliptical dataset does not achieve the global minimum of the inertia.  Because the $k$-means algorithm optimizes over inertia, a successful $k$-means produces a solution with the minimum inertia.  We therefore measure our clustering solutions with the \emph{determinant criterion}, which is the determinant of the within-cluster scatter matrix.\footnote{Inertia is the trace of the within-cluster scatter matrix.}  As we see in Table \ref{Tab:EllipticalPerformance}, while the inertias of the three solutions are similar, the determinant criterion for the correct solution is an order of magnitude smaller.

\begin{figure}
	\subfloat[ ]{
	\includegraphics[width=3.5in]{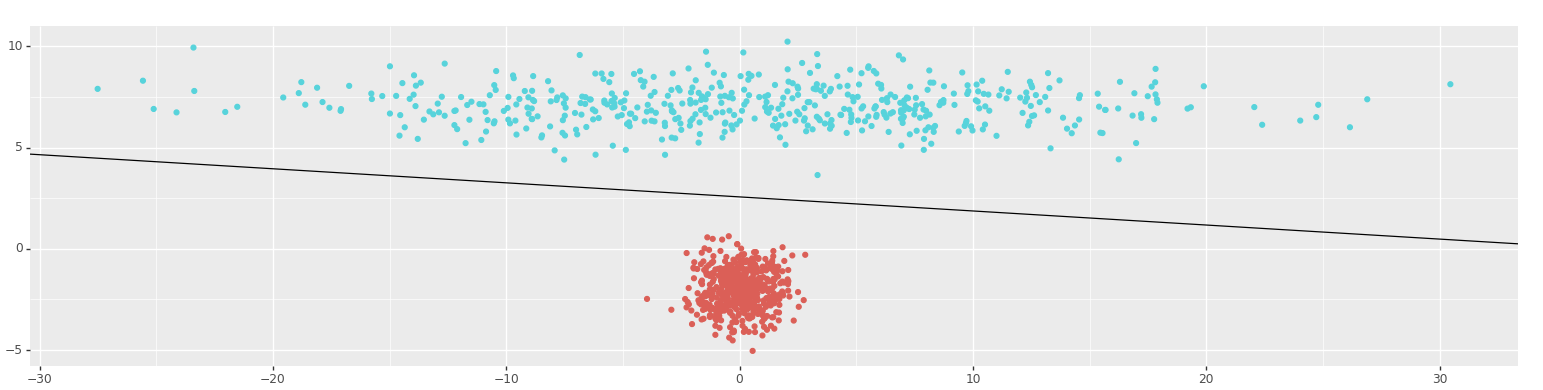}
	}

	\subfloat[ ]{
	\includegraphics[width=3.5in]{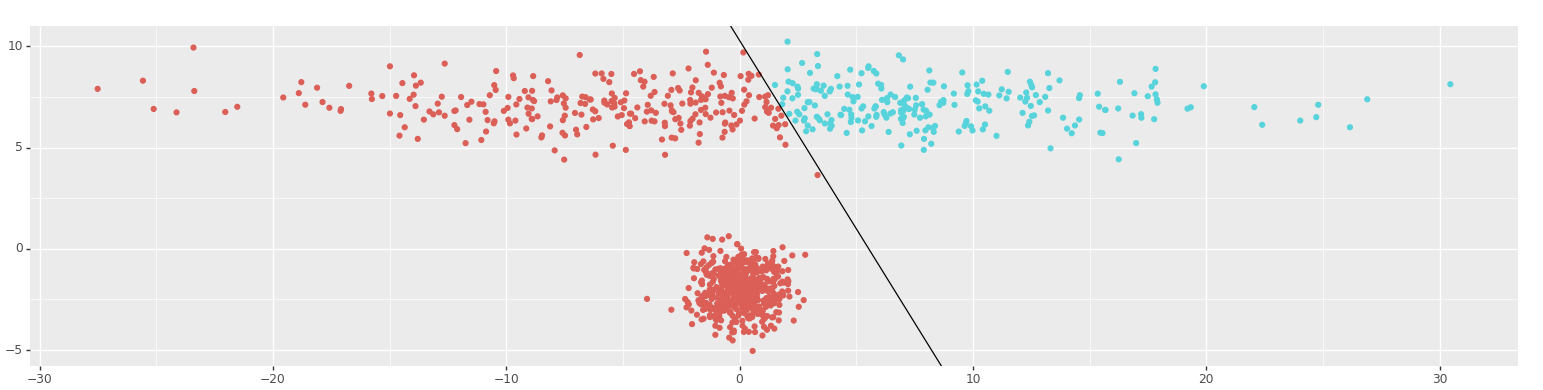}
	}
\caption{The Elliptical dataset clustered (a) correctly, and (b) incorrectly. The second incorrect partition identified in Table \ref{Tab:EllipticalPerformance} is the mirror image of the partition of (b).}%
\label{Fig:Elliptical}
\end{figure}

\begin{table}
\begin{center}
    \begin{tabular}{|c|c|c|}
    \hline
    & \textbf{Inertia}  & $\mathbf{|S_W|}$\\
    \hline
    \textbf{Correct} & \ \ 47,194.898 \ \ & \ \ 44,519,995.441 \ \ \\
    \hline
    \textbf{Incorrect 1} & \ \ 39,503.416 \ \ & \ \ 257,968,951.423 \ \ \\
    \hline 
    \textbf{Incorrect 2} & \ \ 42,489.978 \ \ & \ \ 278,934,320.164 \ \ \\
    \hline
    \end{tabular}
\end{center}%
\caption{Comparison between the inertia and determinant criterion for the three solutions of $k$-means on the elliptical data set, given an optimal choice of $\varepsilon$.}\label{Tab:EllipticalPerformance}%
\end{table}

Because we have used the $\varepsilon$-ball method to map the data to a graph, we show how the choice of $\varepsilon$ affects the performance of $q$-means.  As we can see in Figure \ref{Fig:Epsilons}, there is a wide range of $\varepsilon$'s for which $q$-means has a higher probability of success over both $k$-means and $k$++ for the elliptical dataset. Conversely, on the five cluster dataset, there is a very small range in the choice of $\varepsilon$ that returns a similar success probability as $k$++. Consequently, it is clear that a good estimate for $\varepsilon$ is important, especially on highly idealized cluster data. The reason for wide variance in performance of $q$-means as we vary $\varepsilon$ is that, when $\varepsilon$ is too small, the corresponding graph has more connected components than the number of clusters, whereas when $\varepsilon$ becomes too large, the graph representation of the data does not represent the structure of the vector-valued data. Table \ref{Tab:AlgorithmPerformance} summarizes the success rate and the number of iterations required for each dataset, given an optimal choice of $\varepsilon$.

\begin{figure}
	\subfloat[ ]{
	\includegraphics[width=0.4\textwidth]{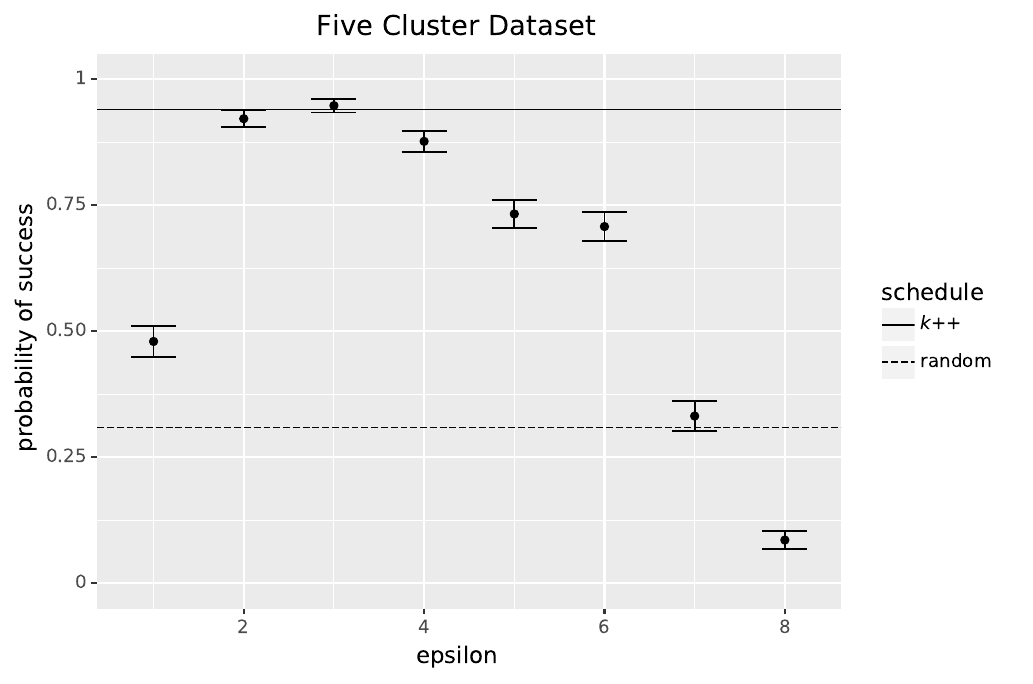}
	}

	\subfloat[ ]{
	\includegraphics[width=0.4\textwidth]{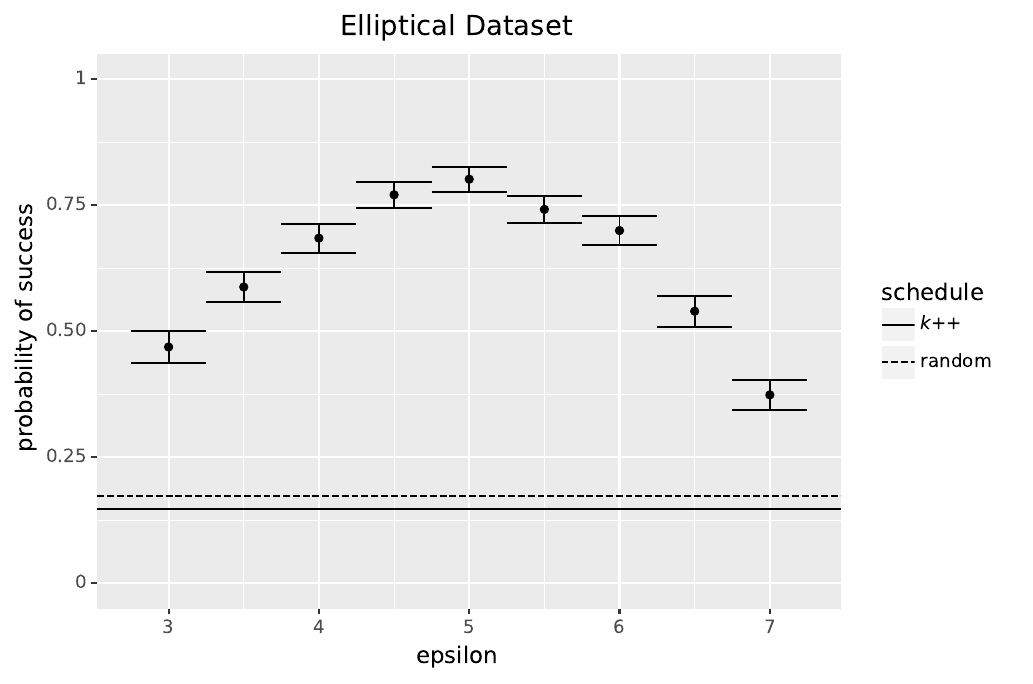}
	}

\caption{A plot of the probabilities of success of $q$-means for various choices of $\varepsilon$ in our construction of $\varepsilon$-balls for (a) the Five Cluster dataset and (b) the Elliptical dataset.}%
\label{Fig:Epsilons}
\end{figure}

\begin{table}
\begin{center}
\begin{tabular}{|c|c|c|c|}
\hline
Algorithm & Success & \# Iterations   & \# Iterations\\
 \ 	  & Rate    & when Successful & when Failed\\
\hline
\multicolumn{4}{|l|}{\textbf{Five Cluster}} \\
\hline
$q$-means & 94.8\% & 1.859$\pm$0.9333 & 4.827$\pm$2.854\\
\hline
$k$++ & 94.0\% & 2.057$\pm$0.625 & 5.319$\pm$2.702\\
\hline
random & 30.9\% & 3.356$\pm$1.038 & 6.233$\pm$2.818\\
\hline
\multicolumn{4}{|l|}{\textbf{Elliptical}}\\
\hline
$q$-means & 80.2\% & 1.559$\pm$0.833 & 9.747$\pm$1.922\\
\hline
$k$++ & 14.8\% & 1.864$\pm$1.092 & 6.484$\pm$2.288\\
\hline
random & 17.3\% & 2.165$\pm$1.017 & 6.826$\pm$2.282\\
\hline
\end{tabular}
\end{center}
\caption{Performance of each of the three initialization variants of $k$-means on the Five Cluster and Elliptical data sets with optimal choice of $\varepsilon$.}\label{Tab:AlgorithmPerformance}%
\end{table}

\subsection{$q$-nn\label{Subsec:qnn}}

The $l$-nearest neighbor algorithm is a simple but effective algorithm for deciding how to associate new data to existing clusters. Suppose we begin with a labeled dataset with each point labeled as being in one of $k$ clusters.  Given a new point $v$, our goal is to assign $v$ to one of the $k$ clusters. The $l$-nearest neighbor algorithm chooses which cluster to associate to $v$ by determining the clusters associated with the $l$ nearest data points.  The point $v$ is assigned to a cluster based on a majority vote of these $l$ points.\footnote{There are different methods for dealing with ties.  Usually though, they are broken arbitrarily.} In what follows, we develop a quantum initialization for the nearest neighbor algorithm, which we call $q$-nn. This initialization uses QCI to make an initial cluster estimation,  and then uses the nearest-neighbor algorithm to partition the rest of the data.


\begin{algorithm}[H]
    \caption{Returns cluster labels for data points}
    \algsetup{indent=2em}
 \begin{algorithmic}[1]
    
    \STATE \textbf{function $q$-nn} $(L, k, m)$\\
    \textbf{Input:} Laplacian $L$, number of clusters $k$, number of measurements for each cluster $m$\\
    \textbf{Output:} label data $\{u_n^1\},\{u_n^2\},\ldots,\{u_n^k\}$\\
    \begin{ALC@g}
    \STATE $u_1^1\Leftarrow$ vertex chosen uniformly at random\\
    \FOR{$i\in [k-1]$}
        \STATE $u_1^{i+1}\Leftarrow \mbox{QCI}\left(L,\{u_1^n\}_{n=1}^{i}\right)$
    \ENDFOR \newline\\
    \FOR{$i \in [m]$}
        \FOR{$j \in [k]$}
            \STATE $u_{i+1}^j \Leftarrow \mbox{QCI}\left(L,\{u_1^i\}_{i\neq j}\right)$
        \ENDFOR\\
    \ENDFOR
    \end{ALC@g}
    \RETURN $\{u_n^1\},\{u_n^2\},\ldots,\{u_n^k\}$
 \end{algorithmic}
 \label{alg:qnn}
\end{algorithm}

We begin with a set of vector data, a value $k$ for the number of clusters, and a value $m$ for the number of measurements we take for each cluster.  We first take a ``thinned'' subset the data and convert this to a graph as described in Section \ref{Sec:NumSim}.  However, we see, even using only a subset of the data provides good results after applying the nearest neighbor algorithm.  The $q$-nn initialization proceeds similarly to that of the $q$-means initialization.  The QCI algorithm is run $k$ times to obtain a representative for each of the $k$ clusters.  Then for the $j$-th cluster, we mark all but the $j$-th representative and run QCI $m$ times.  Once we have done this for each of the clusters, we use a majority vote to determine cluster membership.  Finally, we perform the $l$-nearest neighbor algorithm on any unlabeled data (or in our case the rest of the data set) to finish clustering the data.

The complexity of the quantum subroutine in $q$-nn is dominated by the sampling rate $m$. We need to sample the ground state enough times to obtain a good approximation for the distribution.  The expected value of the $L_1$ error between the empirical distribution and the true distribution can be bounded by $\sqrt{(N-1)/m}$ \cite{Han_2015}.  Therefore, for a given error tolerance, the sampling rate is linear in $N$, giving the quantum subroutine a runtime of $\mc{O}(N)$.  As with $q$-means, the classical cost of building the graph Laplacian dominates the complexity of the overall $q$-nn algorithm.

We compare $q$-nn to Laplacian spectral clustering, which also uses the graph Laplacian and requires an input for the expected number of clusters.  For a true comparison, we test the spectral clustering over subsets of the data and apply the $l$-nearest neighbor algorithm to partition the rest of the data. Since Laplacian spectral clustering works only over connected graphs, thinned subsets that returned disconnected graphs were not considered for spectral clustering.  We apply $q$-nn to two test sets.  The first is a set of connected graphs for comparison against spectral clustering; the second is a set of unrestricted graphs to see the general performance of $q$-nn.

We note that, since $k$-means and $k$++ do not require a graph Laplacian, it was important to factor the additional cost of constructing the graph Laplacian into the complexity of $q$-means when comparing it to $k$-means and $k$++. However, since $q$-nn and spectral clustering both require the graph Laplacian as an input, it is worth comparing the complexity of the two algorithms when the graph Laplacian is assumed to be already provided (e.g. in a graph clustering scenario, as opposed to data clustering). In this case, as already discussed, the cost of the $q$-nn algorithm is dominated by the sampling rate $m$, which has runtime $\mc{O}(N)$. The runtime for spectral clustering, however, is $\mathcal{O}(N^2)$, since it requires computing a subset of the eigenvectors of the the graph Laplacian. Thus, we obtain a quadratic improvement over spectral clustering.

\begin{figure}
	\subfloat[Three Cigars Dataset]{
	\includegraphics[width=0.3\textwidth]{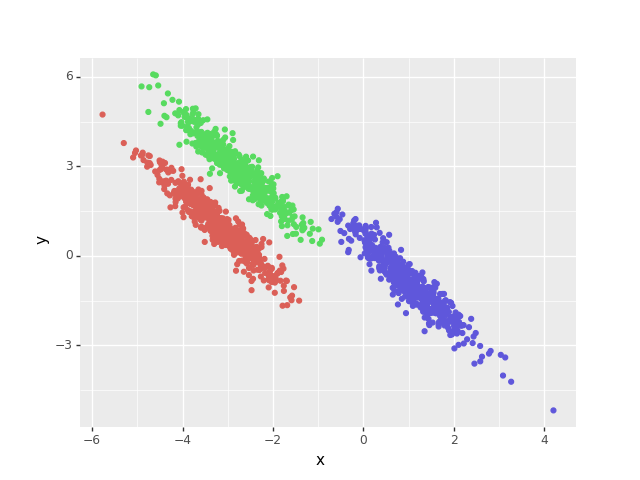}
	}

	\subfloat[Sun and Moon Dataset]{
	\includegraphics[width=0.3\textwidth]{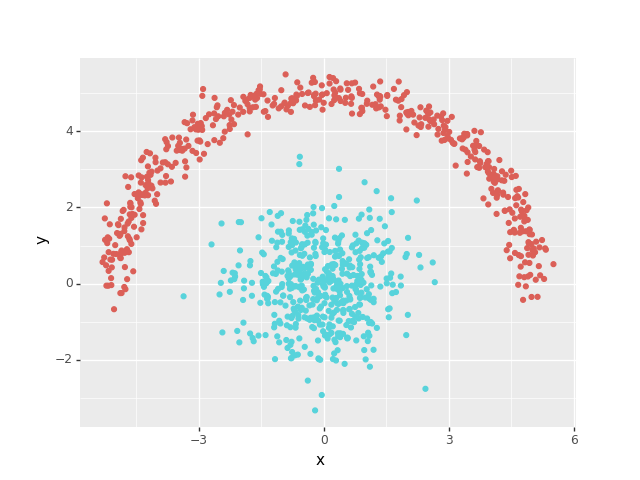}
	}
\caption{Images of the correctly clustered (a) Three Cigars dataset and the (b) Sun and Moon dataset used to test $q$-nn.}%
\label{Fig:qnn-datasets}
\end{figure}

We test $q$-nn on two datasets, shown in Figure \ref{Fig:qnn-datasets}.  The first dataset consists of data drawn from three anisotropic Gaussians (the \emph{Three Cigars}).  The second dataset is a  half circle and a circular Gaussian (the \emph{Sun and Moon}).  As a measure of performance, we use the {\textit{adjusted rand index}} between the cluster solution and the data's true labels.\footnote{The adjusted rand index is a measure of agreement between two partitions.  A score of 1 represents a perfect match between the two partitions while a score of 0 represents a random labeling.} The box plots in Figure \ref{Fig:qnn-boxplots} show the results of $q$-nn and Laplacian spectral clustering for various $\varepsilon$'s.  We see that $q$-nn, when restricted to connected graphs, has a similar performance to unrestricted $q$-nn. We also see that $q$-nn has a median performance similar to spectral clustering on both datasets, albeit with a larger variance.  We believe this variance is caused by the choice of marks, as these choices have a large impact on the distribution over the vertices.

\begin{figure}
	\subfloat[Three Cigars Dataset box plots]{
	\includegraphics[width=0.3\textwidth]{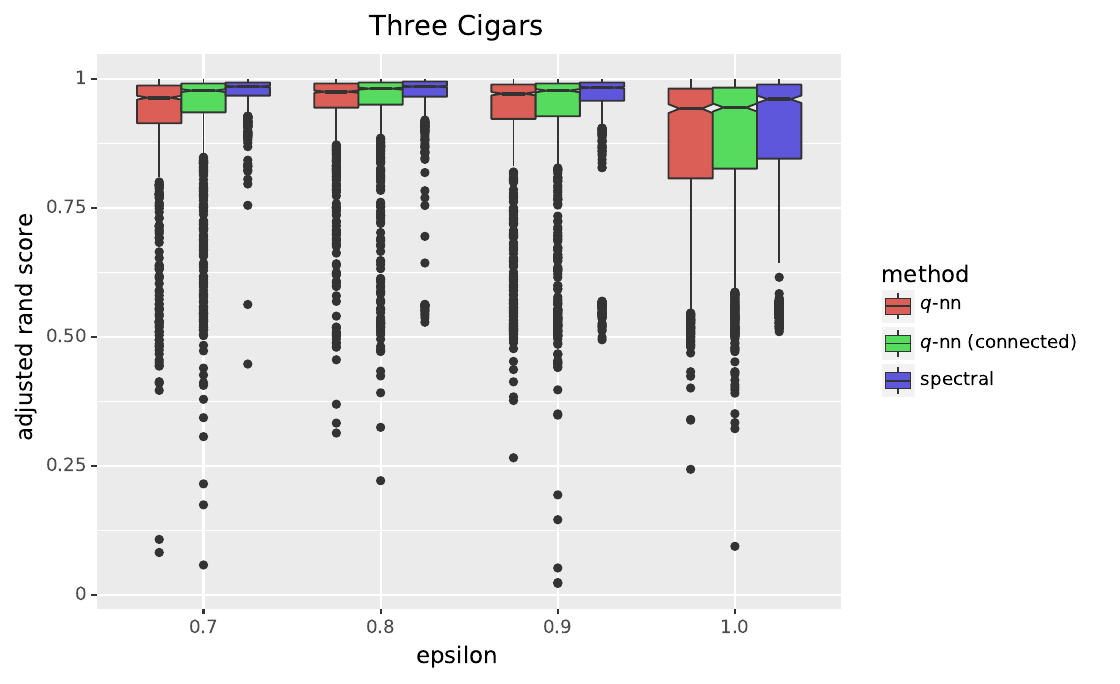}}

	\subfloat[Sun and Moon Dataset box plots]{
	\includegraphics[width=0.3\textwidth]{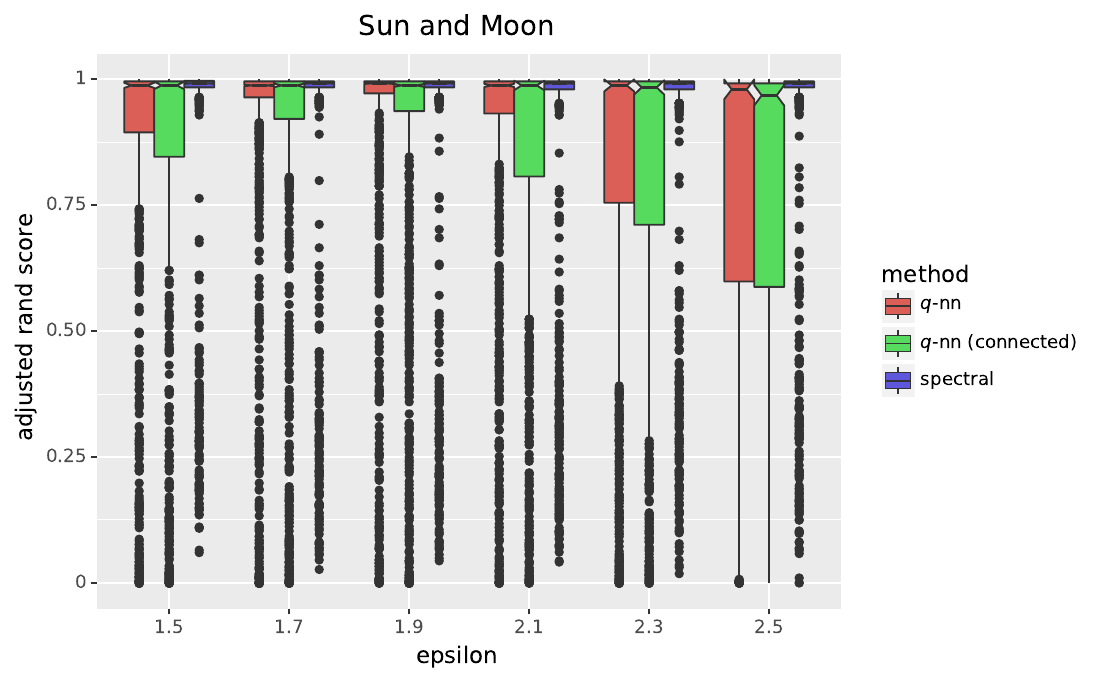}}
\caption{Box plots of $\varepsilon$ vs. adjusted rand scores (ars) for $q$-nn (red) and spectral clustering (blue) for (a) the Three Cigars dataset and (b) the Sun and Moon dataset, over various $\varepsilon$'s.}%
\label{Fig:qnn-boxplots}
\end{figure}

\section{Conclusion}
In conclusion, we have shown that the implicit structure of stoquastic Hamiltonians lends itself well to clustering applications. Taking advantage of this structure, we have developed a quantum cluster initializer and integrated it into new and existing clustering techniques to demonstrate an improvement in various performance parameters. 

As compared to $k$++, we find that our new $q$-means algorithm is applicable to more types of data.  Even as the underlying graph structure departs from the idealized case, $q$-means will choose seeds that produce good results after the $k$-means algorithm is applied.  While there is more computation cost up front in building the graph Laplacian, the quantum subroutine runs in constant time.  In addition, empirically, we see fewer iterations required to converge to a correct solution on our datasets.

As compared to spectral clustering, we find that our $q$-nn algorithm has a similar success rate for all datasets tested, albeit with a slightly larger variance. Unlike spectral clustering, however, our $q$-nn algorithm does not require the underlying graph to be connected; and in fact, the peformance does not seem to be affected by the restriction to connected graphs.  Moreover, $q$-nn gives a speedup in runtime over spectral clustering, when the data is presented as a graph.  When the data is vector-valued, the two algorithms have comparable runtime due to a bottleneck in creating the graph.

In an effort to demonstrate utility on near-term devices, we have assumed that the total number of qubits is small. Even though the number of qubits required is logarithmic in the size of the data ($N$ vertices are represented with $\log_2(N)$ qubits), we have implicitly assumed all-to-all connectivity between these qubits. This assumption is extremely hard to satisfy in practice. Practically, even after using clever embedding techniques, an adiabatic quantum computer may only have an extremely small number of fully-connected qubits. As such, in our simulations, we experimented with ``thinning'' the data for the quantum subroutine. Namely, we selected uniformly random subsets of various sizes from the full dataset from which to build the problem Hamiltonians, and use these solutions as the inputs into the classical subroutines. Our analysis shows empirically that even when the quantum subroutine was run over these smaller sample sizes, it still generated useful inputs for the classical routines. Consequently, we expect that these algorithms should be implementable on near-term quantum devices.

\begin{acknowledgments}
This project was funded through Naval Surface Warfare Center, Dahlgren Division (NSWCDD) Navy Innovative Science and Engineering (NISE) program. The authors are grateful for useful discussions and insights provided by Dave Marchette, the help from Matthew Crawford, and the rest of the members of the AM\&DA group. The authors would additionally like to thank an anonymous referee, whose comments and suggestions greatly improved the overall quality of this paper.
\end{acknowledgments}

\bibliography{NISQ2}

\begin{thebibliography}{28}%
\makeatletter
\providecommand \@ifxundefined [1]{%
 \@ifx{#1\undefined}
}%
\providecommand \@ifnum [1]{%
 \ifnum #1\expandafter \@firstoftwo
 \else \expandafter \@secondoftwo
 \fi
}%
\providecommand \@ifx [1]{%
 \ifx #1\expandafter \@firstoftwo
 \else \expandafter \@secondoftwo
 \fi
}%
\providecommand \natexlab [1]{#1}%
\providecommand \enquote  [1]{``#1''}%
\providecommand \bibnamefont  [1]{#1}%
\providecommand \bibfnamefont [1]{#1}%
\providecommand \citenamefont [1]{#1}%
\providecommand \href@noop [0]{\@secondoftwo}%
\providecommand \href [0]{\begingroup \@sanitize@url \@href}%
\providecommand \@href[1]{\@@startlink{#1}\@@href}%
\providecommand \@@href[1]{\endgroup#1\@@endlink}%
\providecommand \@sanitize@url [0]{\catcode `\\12\catcode `\$12\catcode
  `\&12\catcode `\#12\catcode `\^12\catcode `\_12\catcode `\%12\relax}%
\providecommand \@@startlink[1]{}%
\providecommand \@@endlink[0]{}%
\providecommand \url  [0]{\begingroup\@sanitize@url \@url }%
\providecommand \@url [1]{\endgroup\@href {#1}{\urlprefix }}%
\providecommand \urlprefix  [0]{URL }%
\providecommand \Eprint [0]{\href }%
\providecommand \doibase [0]{http://dx.doi.org/}%
\providecommand \selectlanguage [0]{\@gobble}%
\providecommand \bibinfo  [0]{\@secondoftwo}%
\providecommand \bibfield  [0]{\@secondoftwo}%
\providecommand \translation [1]{[#1]}%
\providecommand \BibitemOpen [0]{}%
\providecommand \bibitemStop [0]{}%
\providecommand \bibitemNoStop [0]{.\EOS\space}%
\providecommand \EOS [0]{\spacefactor3000\relax}%
\providecommand \BibitemShut  [1]{\csname bibitem#1\endcsname}%
\let\auto@bib@innerbib\@empty
\bibitem [{\citenamefont {Preskill}(2018)}]{Preskill_2018}%
  \BibitemOpen
  \bibfield  {author} {\bibinfo {author} {\bibfnamefont {J.}~\bibnamefont
  {Preskill}},\ }\href {\doibase 10.22331/q-2018-08-06-79} {\bibfield
  {journal} {\bibinfo  {journal} {{Quantum}}\ }\textbf {\bibinfo {volume}
  {2}},\ \bibinfo {pages} {79} (\bibinfo {year} {2018})}\BibitemShut {NoStop}%
\bibitem [{\citenamefont {Dunjko}\ \emph {et~al.}(2018)\citenamefont {Dunjko},
  \citenamefont {Ge},\ and\ \citenamefont {Cirac}}]{DunjkoEtAl_2018}%
  \BibitemOpen
  \bibfield  {author} {\bibinfo {author} {\bibfnamefont {V.}~\bibnamefont
  {Dunjko}}, \bibinfo {author} {\bibfnamefont {Y.}~\bibnamefont {Ge}}, \ and\
  \bibinfo {author} {\bibfnamefont {J.~I.}\ \bibnamefont {Cirac}},\ }\href
  {\doibase 10.1103/PhysRevLett.121.250501} {\bibfield  {journal} {\bibinfo
  {journal} {Phys. Rev. Lett.}\ }\textbf {\bibinfo {volume} {121}},\ \bibinfo
  {pages} {250501} (\bibinfo {year} {2018})}\BibitemShut {NoStop}%
\bibitem [{\citenamefont {Farhi}\ and\ \citenamefont
  {Neven}(2018)}]{FarhiEtAl_2018}%
  \BibitemOpen
  \bibfield  {author} {\bibinfo {author} {\bibfnamefont {E.}~\bibnamefont
  {Farhi}}\ and\ \bibinfo {author} {\bibfnamefont {H.}~\bibnamefont {Neven}},\
  }\href@noop {} {\enquote {\bibinfo {title} {Classification with quantum
  neural networks on near term processors},}\ } (\bibinfo {year} {2018}),\
  \Eprint {http://arxiv.org/abs/arXiv:1802.06002} {arXiv:1802.06002}
  \BibitemShut {NoStop}%
\bibitem [{\citenamefont {Farhi}\ \emph {et~al.}(2017)\citenamefont {Farhi},
  \citenamefont {Goldstone}, \citenamefont {Gutmann},\ and\ \citenamefont
  {Neven}}]{FarhiEtAl_2017}%
  \BibitemOpen
  \bibfield  {author} {\bibinfo {author} {\bibfnamefont {E.}~\bibnamefont
  {Farhi}}, \bibinfo {author} {\bibfnamefont {J.}~\bibnamefont {Goldstone}},
  \bibinfo {author} {\bibfnamefont {S.}~\bibnamefont {Gutmann}}, \ and\
  \bibinfo {author} {\bibfnamefont {H.}~\bibnamefont {Neven}},\ }\href@noop {}
  {\enquote {\bibinfo {title} {Quantum algorithms for fixed qubit
  architectures},}\ } (\bibinfo {year} {2017}),\ \Eprint
  {http://arxiv.org/abs/arXiv:1703.06199} {arXiv:1703.06199} \BibitemShut
  {NoStop}%
\bibitem [{\citenamefont {Kandala}\ \emph {et~al.}(2017)\citenamefont
  {Kandala}, \citenamefont {Mezzacapo}, \citenamefont {Temme}, \citenamefont
  {Takita}, \citenamefont {Brink}, \citenamefont {Chow},\ and\ \citenamefont
  {Gambetta}}]{Kandala2017}%
  \BibitemOpen
  \bibfield  {author} {\bibinfo {author} {\bibfnamefont {A.}~\bibnamefont
  {Kandala}}, \bibinfo {author} {\bibfnamefont {A.}~\bibnamefont {Mezzacapo}},
  \bibinfo {author} {\bibfnamefont {K.}~\bibnamefont {Temme}}, \bibinfo
  {author} {\bibfnamefont {M.}~\bibnamefont {Takita}}, \bibinfo {author}
  {\bibfnamefont {M.}~\bibnamefont {Brink}}, \bibinfo {author} {\bibfnamefont
  {J.~M.}\ \bibnamefont {Chow}}, \ and\ \bibinfo {author} {\bibfnamefont
  {J.~M.}\ \bibnamefont {Gambetta}},\ }\href
  {https://doi.org/10.1038/nature23879} {\bibfield  {journal} {\bibinfo
  {journal} {Nature}\ }\textbf {\bibinfo {volume} {549}},\ \bibinfo {pages}
  {242 EP } (\bibinfo {year} {2017})}\BibitemShut {NoStop}%
\bibitem [{\citenamefont {Huggins}\ \emph {et~al.}(2019)\citenamefont
  {Huggins}, \citenamefont {Patil}, \citenamefont {Mitchell}, \citenamefont
  {Whaley},\ and\ \citenamefont {Stoudenmire}}]{Huggins_2019}%
  \BibitemOpen
  \bibfield  {author} {\bibinfo {author} {\bibfnamefont {W.}~\bibnamefont
  {Huggins}}, \bibinfo {author} {\bibfnamefont {P.}~\bibnamefont {Patil}},
  \bibinfo {author} {\bibfnamefont {B.}~\bibnamefont {Mitchell}}, \bibinfo
  {author} {\bibfnamefont {K.~B.}\ \bibnamefont {Whaley}}, \ and\ \bibinfo
  {author} {\bibfnamefont {E.~M.}\ \bibnamefont {Stoudenmire}},\ }\href
  {\doibase 10.1088/2058-9565/aaea94} {\bibfield  {journal} {\bibinfo
  {journal} {Quantum Science and Technology}\ }\textbf {\bibinfo {volume}
  {4}},\ \bibinfo {pages} {024001} (\bibinfo {year} {2019})}\BibitemShut
  {NoStop}%
\bibitem [{\citenamefont {Havlicek}\ \emph {et~al.}(2018)\citenamefont
  {Havlicek}, \citenamefont {Córcoles}, \citenamefont {Temme}, \citenamefont
  {Harrow}, \citenamefont {Kandala}, \citenamefont {Chow},\ and\ \citenamefont
  {Gambetta}}]{HavlicekEtAl_2018}%
  \BibitemOpen
  \bibfield  {author} {\bibinfo {author} {\bibfnamefont {V.}~\bibnamefont
  {Havlicek}}, \bibinfo {author} {\bibfnamefont {A.~D.}\ \bibnamefont
  {Córcoles}}, \bibinfo {author} {\bibfnamefont {K.}~\bibnamefont {Temme}},
  \bibinfo {author} {\bibfnamefont {A.~W.}\ \bibnamefont {Harrow}}, \bibinfo
  {author} {\bibfnamefont {A.}~\bibnamefont {Kandala}}, \bibinfo {author}
  {\bibfnamefont {J.~M.}\ \bibnamefont {Chow}}, \ and\ \bibinfo {author}
  {\bibfnamefont {J.~M.}\ \bibnamefont {Gambetta}},\ }\href@noop {} {\enquote
  {\bibinfo {title} {Supervised learning with quantum enhanced feature
  spaces},}\ } (\bibinfo {year} {2018}),\ \Eprint
  {http://arxiv.org/abs/arXiv:1804.11326} {arXiv:1804.11326} \BibitemShut
  {NoStop}%
\bibitem [{\citenamefont {Schuld}\ \emph {et~al.}(2018)\citenamefont {Schuld},
  \citenamefont {Bocharov}, \citenamefont {Svore},\ and\ \citenamefont
  {Wiebe}}]{SchuldEtAl_2018}%
  \BibitemOpen
  \bibfield  {author} {\bibinfo {author} {\bibfnamefont {M.}~\bibnamefont
  {Schuld}}, \bibinfo {author} {\bibfnamefont {A.}~\bibnamefont {Bocharov}},
  \bibinfo {author} {\bibfnamefont {K.}~\bibnamefont {Svore}}, \ and\ \bibinfo
  {author} {\bibfnamefont {N.}~\bibnamefont {Wiebe}},\ }\href@noop {} {\enquote
  {\bibinfo {title} {Circuit-centric quantum classifiers},}\ } (\bibinfo {year}
  {2018}),\ \Eprint {http://arxiv.org/abs/arXiv:1804.00633} {arXiv:1804.00633}
  \BibitemShut {NoStop}%
\bibitem [{\citenamefont {Cross}\ \emph {et~al.}(2015)\citenamefont {Cross},
  \citenamefont {Smith},\ and\ \citenamefont {Smolin}}]{CrossEtAl_2015}%
  \BibitemOpen
  \bibfield  {author} {\bibinfo {author} {\bibfnamefont {A.~W.}\ \bibnamefont
  {Cross}}, \bibinfo {author} {\bibfnamefont {G.}~\bibnamefont {Smith}}, \ and\
  \bibinfo {author} {\bibfnamefont {J.~A.}\ \bibnamefont {Smolin}},\ }\href
  {\doibase 10.1103/PhysRevA.92.012327} {\bibfield  {journal} {\bibinfo
  {journal} {Phys. Rev. A}\ }\textbf {\bibinfo {volume} {92}},\ \bibinfo
  {pages} {012327} (\bibinfo {year} {2015})}\BibitemShut {NoStop}%
\bibitem [{\citenamefont {Wilson}\ \emph {et~al.}(2018)\citenamefont {Wilson},
  \citenamefont {Otterbach}, \citenamefont {Tezak}, \citenamefont {Smith},
  \citenamefont {Crooks},\ and\ \citenamefont {da~Silva}}]{WilsonEtAl_2018}%
  \BibitemOpen
  \bibfield  {author} {\bibinfo {author} {\bibfnamefont {C.~M.}\ \bibnamefont
  {Wilson}}, \bibinfo {author} {\bibfnamefont {J.~S.}\ \bibnamefont
  {Otterbach}}, \bibinfo {author} {\bibfnamefont {N.}~\bibnamefont {Tezak}},
  \bibinfo {author} {\bibfnamefont {R.~S.}\ \bibnamefont {Smith}}, \bibinfo
  {author} {\bibfnamefont {G.~E.}\ \bibnamefont {Crooks}}, \ and\ \bibinfo
  {author} {\bibfnamefont {M.~P.}\ \bibnamefont {da~Silva}},\ }\href@noop {}
  {\enquote {\bibinfo {title} {Quantum kitchen sinks: An algorithm for machine
  learning on near-term quantum computers},}\ } (\bibinfo {year} {2018}),\
  \Eprint {http://arxiv.org/abs/arXiv:1806.08321} {arXiv:1806.08321}
  \BibitemShut {NoStop}%
\bibitem [{\citenamefont {Farhi}\ \emph
  {et~al.}(2014{\natexlab{a}})\citenamefont {Farhi}, \citenamefont
  {Goldstone},\ and\ \citenamefont {Gutmann}}]{FarhiEtAl_2014}%
  \BibitemOpen
  \bibfield  {author} {\bibinfo {author} {\bibfnamefont {E.}~\bibnamefont
  {Farhi}}, \bibinfo {author} {\bibfnamefont {J.}~\bibnamefont {Goldstone}}, \
  and\ \bibinfo {author} {\bibfnamefont {S.}~\bibnamefont {Gutmann}},\
  }\href@noop {} {\enquote {\bibinfo {title} {A quantum approximate
  optimization algorithm},}\ } (\bibinfo {year} {2014}{\natexlab{a}}),\ \Eprint
  {http://arxiv.org/abs/arXiv:1411.4028} {arXiv:1411.4028} \BibitemShut
  {NoStop}%
\bibitem [{\citenamefont {Farhi}\ \emph
  {et~al.}(2014{\natexlab{b}})\citenamefont {Farhi}, \citenamefont
  {Goldstone},\ and\ \citenamefont {Gutmann}}]{FarhiEtAl_2014-2}%
  \BibitemOpen
  \bibfield  {author} {\bibinfo {author} {\bibfnamefont {E.}~\bibnamefont
  {Farhi}}, \bibinfo {author} {\bibfnamefont {J.}~\bibnamefont {Goldstone}}, \
  and\ \bibinfo {author} {\bibfnamefont {S.}~\bibnamefont {Gutmann}},\
  }\href@noop {} {\enquote {\bibinfo {title} {A quantum approximate
  optimization algorithm applied to a bounded occurrence constraint problem},}\
  } (\bibinfo {year} {2014}{\natexlab{b}}),\ \Eprint
  {http://arxiv.org/abs/arXiv:1412.6062} {arXiv:1412.6062} \BibitemShut
  {NoStop}%
\bibitem [{\citenamefont {Farhi}\ and\ \citenamefont
  {Harrow}(2016)}]{FarhiHarrow_2016}%
  \BibitemOpen
  \bibfield  {author} {\bibinfo {author} {\bibfnamefont {E.}~\bibnamefont
  {Farhi}}\ and\ \bibinfo {author} {\bibfnamefont {A.~W.}\ \bibnamefont
  {Harrow}},\ }\href@noop {} {\enquote {\bibinfo {title} {Quantum supremacy
  through the quantum approximate optimization algorithm},}\ } (\bibinfo {year}
  {2016}),\ \Eprint {http://arxiv.org/abs/arXiv:1602.07674} {arXiv:1602.07674}
  \BibitemShut {NoStop}%
\bibitem [{\citenamefont {Perusso}\ \emph {et~al.}(2014)\citenamefont
  {Perusso}, \citenamefont {McClean}, \citenamefont {Shadbolt}, \citenamefont
  {Yung}, \citenamefont {Zhou}, \citenamefont {Love}, \citenamefont
  {Asuru-Guzik},\ and\ \citenamefont {O'Brien}}]{Peruzzo_2014}%
  \BibitemOpen
  \bibfield  {author} {\bibinfo {author} {\bibfnamefont {A.}~\bibnamefont
  {Perusso}}, \bibinfo {author} {\bibfnamefont {J.}~\bibnamefont {McClean}},
  \bibinfo {author} {\bibfnamefont {P.}~\bibnamefont {Shadbolt}}, \bibinfo
  {author} {\bibfnamefont {M.-H.}\ \bibnamefont {Yung}}, \bibinfo {author}
  {\bibfnamefont {X.-Q.}\ \bibnamefont {Zhou}}, \bibinfo {author}
  {\bibfnamefont {P.~J.}\ \bibnamefont {Love}}, \bibinfo {author}
  {\bibfnamefont {A.}~\bibnamefont {Asuru-Guzik}}, \ and\ \bibinfo {author}
  {\bibfnamefont {J.}~\bibnamefont {O'Brien}},\ }\href {\doibase
  10.1038/ncomms5213} {\bibfield  {journal} {\bibinfo  {journal} {Nature
  Communications}\ }\textbf {\bibinfo {volume} {5}} (\bibinfo {year} {2014}),\
  10.1038/ncomms5213}\BibitemShut {NoStop}%
\bibitem [{\citenamefont {McClean}\ \emph {et~al.}(2016)\citenamefont
  {McClean}, \citenamefont {Romero}, \citenamefont {Babbush},\ and\
  \citenamefont {Aspuru-Guzik}}]{McClean_2016}%
  \BibitemOpen
  \bibfield  {author} {\bibinfo {author} {\bibfnamefont {J.~R.}\ \bibnamefont
  {McClean}}, \bibinfo {author} {\bibfnamefont {J.}~\bibnamefont {Romero}},
  \bibinfo {author} {\bibfnamefont {R.}~\bibnamefont {Babbush}}, \ and\
  \bibinfo {author} {\bibfnamefont {A.}~\bibnamefont {Aspuru-Guzik}},\ }\href
  {\doibase 10.1088/1367-2630/18/2/023023} {\bibfield  {journal} {\bibinfo
  {journal} {New Journal of Physics}\ }\textbf {\bibinfo {volume} {18}},\
  \bibinfo {pages} {023023} (\bibinfo {year} {2016})}\BibitemShut {NoStop}%
\bibitem [{\citenamefont {Biamonte}\ \emph {et~al.}(2017)\citenamefont
  {Biamonte}, \citenamefont {Wittek}, \citenamefont {Pancotti}, \citenamefont
  {Rebentrost}, \citenamefont {Wiebe},\ and\ \citenamefont
  {Lloyd}}]{BiamonteEtAl_2017}%
  \BibitemOpen
  \bibfield  {author} {\bibinfo {author} {\bibfnamefont {J.}~\bibnamefont
  {Biamonte}}, \bibinfo {author} {\bibfnamefont {P.}~\bibnamefont {Wittek}},
  \bibinfo {author} {\bibfnamefont {N.}~\bibnamefont {Pancotti}}, \bibinfo
  {author} {\bibfnamefont {P.}~\bibnamefont {Rebentrost}}, \bibinfo {author}
  {\bibfnamefont {N.}~\bibnamefont {Wiebe}}, \ and\ \bibinfo {author}
  {\bibfnamefont {S.}~\bibnamefont {Lloyd}},\ }\href
  {https://doi.org/10.1038/nature23474} {\bibfield  {journal} {\bibinfo
  {journal} {Nature}\ }\textbf {\bibinfo {volume} {549}},\ \bibinfo {pages}
  {195 EP } (\bibinfo {year} {2017})}\BibitemShut {NoStop}%
\bibitem [{\citenamefont {Fujii}(2018)}]{Fujii_2018}%
  \BibitemOpen
  \bibfield  {author} {\bibinfo {author} {\bibfnamefont {K.}~\bibnamefont
  {Fujii}},\ }\href@noop {} {\enquote {\bibinfo {title} {Quantum speedup in
  stoquastic adiabatic quantum computation},}\ } (\bibinfo {year} {2018}),\
  \Eprint {http://arxiv.org/abs/arXiv:1803.09954} {arXiv:1803.09954}
  \BibitemShut {NoStop}%
\bibitem [{\citenamefont {Hastings}(2013)}]{Hastings_2013}%
  \BibitemOpen
  \bibfield  {author} {\bibinfo {author} {\bibfnamefont {M.~B.}\ \bibnamefont
  {Hastings}},\ }\href {http://dl.acm.org/citation.cfm?id=2535639.2535647}
  {\bibfield  {journal} {\bibinfo  {journal} {Quantum Info. Comput.}\ }\textbf
  {\bibinfo {volume} {13}},\ \bibinfo {pages} {1038} (\bibinfo {year}
  {2013})}\BibitemShut {NoStop}%
\bibitem [{\citenamefont {Jarret}\ \emph {et~al.}(2016)\citenamefont {Jarret},
  \citenamefont {Jordan},\ and\ \citenamefont {Lackey}}]{JarretEtAl_2016}%
  \BibitemOpen
  \bibfield  {author} {\bibinfo {author} {\bibfnamefont {M.}~\bibnamefont
  {Jarret}}, \bibinfo {author} {\bibfnamefont {S.~P.}\ \bibnamefont {Jordan}},
  \ and\ \bibinfo {author} {\bibfnamefont {B.}~\bibnamefont {Lackey}},\ }\href
  {\doibase 10.1103/PhysRevA.94.042318} {\bibfield  {journal} {\bibinfo
  {journal} {Phys. Rev. A}\ }\textbf {\bibinfo {volume} {94}},\ \bibinfo
  {pages} {042318} (\bibinfo {year} {2016})}\BibitemShut {NoStop}%
\bibitem [{\citenamefont {Jarret}(2018)}]{Jarret_2018}%
  \BibitemOpen
  \bibfield  {author} {\bibinfo {author} {\bibfnamefont {M.}~\bibnamefont
  {Jarret}},\ }\href@noop {} {\enquote {\bibinfo {title} {Hamiltonian surgery:
  Cheeger-type gap inequalities for nonpositive (stoquastic), real, and
  hermitian matrices},}\ } (\bibinfo {year} {2018}),\ \Eprint
  {http://arxiv.org/abs/arXiv:1804.06857} {arXiv:1804.06857} \BibitemShut
  {NoStop}%
\bibitem [{\citenamefont {Chung}(1997)}]{Chung_1997}%
  \BibitemOpen
  \bibfield  {author} {\bibinfo {author} {\bibfnamefont {F.~R.~K.}\
  \bibnamefont {Chung}},\ }\href@noop {} {\emph {\bibinfo {title} {Spectral
  Graph Theory}}}\ (\bibinfo  {publisher} {American Mathematical Society},\
  \bibinfo {year} {1997})\BibitemShut {NoStop}%
\bibitem [{\citenamefont {Avron}\ \emph {et~al.}(2010)\citenamefont {Avron},
  \citenamefont {Fraas}, \citenamefont {Graf},\ and\ \citenamefont
  {Grech}}]{AvronEtAl_2010}%
  \BibitemOpen
  \bibfield  {author} {\bibinfo {author} {\bibfnamefont {J.~E.}\ \bibnamefont
  {Avron}}, \bibinfo {author} {\bibfnamefont {M.}~\bibnamefont {Fraas}},
  \bibinfo {author} {\bibfnamefont {G.~M.}\ \bibnamefont {Graf}}, \ and\
  \bibinfo {author} {\bibfnamefont {P.}~\bibnamefont {Grech}},\ }\href
  {\doibase 10.1103/PhysRevA.82.040304} {\bibfield  {journal} {\bibinfo
  {journal} {Phys. Rev. A}\ }\textbf {\bibinfo {volume} {82}},\ \bibinfo
  {pages} {040304} (\bibinfo {year} {2010})}\BibitemShut {NoStop}%
\bibitem [{\citenamefont {Jarret}\ and\ \citenamefont
  {Lackey}(2017)}]{JarretLackey_2017}%
  \BibitemOpen
  \bibfield  {author} {\bibinfo {author} {\bibfnamefont {M.}~\bibnamefont
  {Jarret}}\ and\ \bibinfo {author} {\bibfnamefont {B.}~\bibnamefont
  {Lackey}},\ }\href@noop {} {\enquote {\bibinfo {title} {Substochastic monte
  carlo algorithms},}\ } (\bibinfo {year} {2017}),\ \Eprint
  {http://arxiv.org/abs/arXiv:1704.09014} {arXiv:1704.09014} \BibitemShut
  {NoStop}%
\bibitem [{\citenamefont {Daitch}\ \emph {et~al.}(2009)\citenamefont {Daitch},
  \citenamefont {Kelner},\ and\ \citenamefont {Spielman}}]{DaitchEtAl_2009}%
  \BibitemOpen
  \bibfield  {author} {\bibinfo {author} {\bibfnamefont {S.~I.}\ \bibnamefont
  {Daitch}}, \bibinfo {author} {\bibfnamefont {J.~A.}\ \bibnamefont {Kelner}},
  \ and\ \bibinfo {author} {\bibfnamefont {D.~A.}\ \bibnamefont {Spielman}},\
  }in\ \href {\doibase 10.1145/1553374.1553400} {\emph {\bibinfo {booktitle}
  {Proceedings of the 26th Annual International Conference on Machine
  Learning}}},\ \bibinfo {series and number} {ICML '09}\ (\bibinfo  {publisher}
  {ACM},\ \bibinfo {address} {New York, NY, USA},\ \bibinfo {year} {2009})\
  pp.\ \bibinfo {pages} {201--208}\BibitemShut {NoStop}%
\bibitem [{\citenamefont {{Lloyd}}(1982)}]{Lloyd_1982}%
  \BibitemOpen
  \bibfield  {author} {\bibinfo {author} {\bibfnamefont {S.}~\bibnamefont
  {{Lloyd}}},\ }\href {\doibase 10.1109/TIT.1982.1056489} {\bibfield  {journal}
  {\bibinfo  {journal} {IEEE Transactions on Information Theory}\ }\textbf
  {\bibinfo {volume} {28}},\ \bibinfo {pages} {129} (\bibinfo {year}
  {1982})}\BibitemShut {NoStop}%
\bibitem [{\citenamefont {Arthur}\ and\ \citenamefont
  {Vassilvitskii}(2007)}]{Arthur_2007}%
  \BibitemOpen
  \bibfield  {author} {\bibinfo {author} {\bibfnamefont {D.}~\bibnamefont
  {Arthur}}\ and\ \bibinfo {author} {\bibfnamefont {S.}~\bibnamefont
  {Vassilvitskii}},\ }in\ \href
  {http://dl.acm.org/citation.cfm?id=1283383.1283494} {\emph {\bibinfo
  {booktitle} {Proceedings of the Eighteenth Annual ACM-SIAM Symposium on
  Discrete Algorithms}}},\ \bibinfo {series and number} {SODA '07}\ (\bibinfo
  {publisher} {Society for Industrial and Applied Mathematics},\ \bibinfo
  {address} {Philadelphia, PA, USA},\ \bibinfo {year} {2007})\ pp.\ \bibinfo
  {pages} {1027--1035}\BibitemShut {NoStop}%
\bibitem [{\citenamefont {Lloyd}\ \emph {et~al.}(2013)\citenamefont {Lloyd},
  \citenamefont {Mohseni},\ and\ \citenamefont {Rebentrost}}]{LloydEtAl_2013}%
  \BibitemOpen
  \bibfield  {author} {\bibinfo {author} {\bibfnamefont {S.}~\bibnamefont
  {Lloyd}}, \bibinfo {author} {\bibfnamefont {M.}~\bibnamefont {Mohseni}}, \
  and\ \bibinfo {author} {\bibfnamefont {P.}~\bibnamefont {Rebentrost}},\
  }\href@noop {} {\enquote {\bibinfo {title} {Quantum algorithms for supervised
  and unsupervised machine learning},}\ } (\bibinfo {year} {2013}),\ \Eprint
  {http://arxiv.org/abs/arXiv:1307.0411} {arXiv:1307.0411} \BibitemShut
  {NoStop}%
\bibitem [{\citenamefont {{Han}}\ \emph {et~al.}(2015)\citenamefont {{Han}},
  \citenamefont {{Jiao}},\ and\ \citenamefont {{Weissman}}}]{Han_2015}%
  \BibitemOpen
  \bibfield  {author} {\bibinfo {author} {\bibfnamefont {Y.}~\bibnamefont
  {{Han}}}, \bibinfo {author} {\bibfnamefont {J.}~\bibnamefont {{Jiao}}}, \
  and\ \bibinfo {author} {\bibfnamefont {T.}~\bibnamefont {{Weissman}}},\
  }\href {\doibase 10.1109/TIT.2015.2478816} {\bibfield  {journal} {\bibinfo
  {journal} {IEEE Transactions on Information Theory}\ }\textbf {\bibinfo
  {volume} {61}},\ \bibinfo {pages} {6343} (\bibinfo {year}
  {2015})}\BibitemShut {NoStop}%
\end{thebibliography}%

\appendix*\label{appendix}
\section{}

\begin{figure}
    \centering
    \includegraphics[width=3.5in]{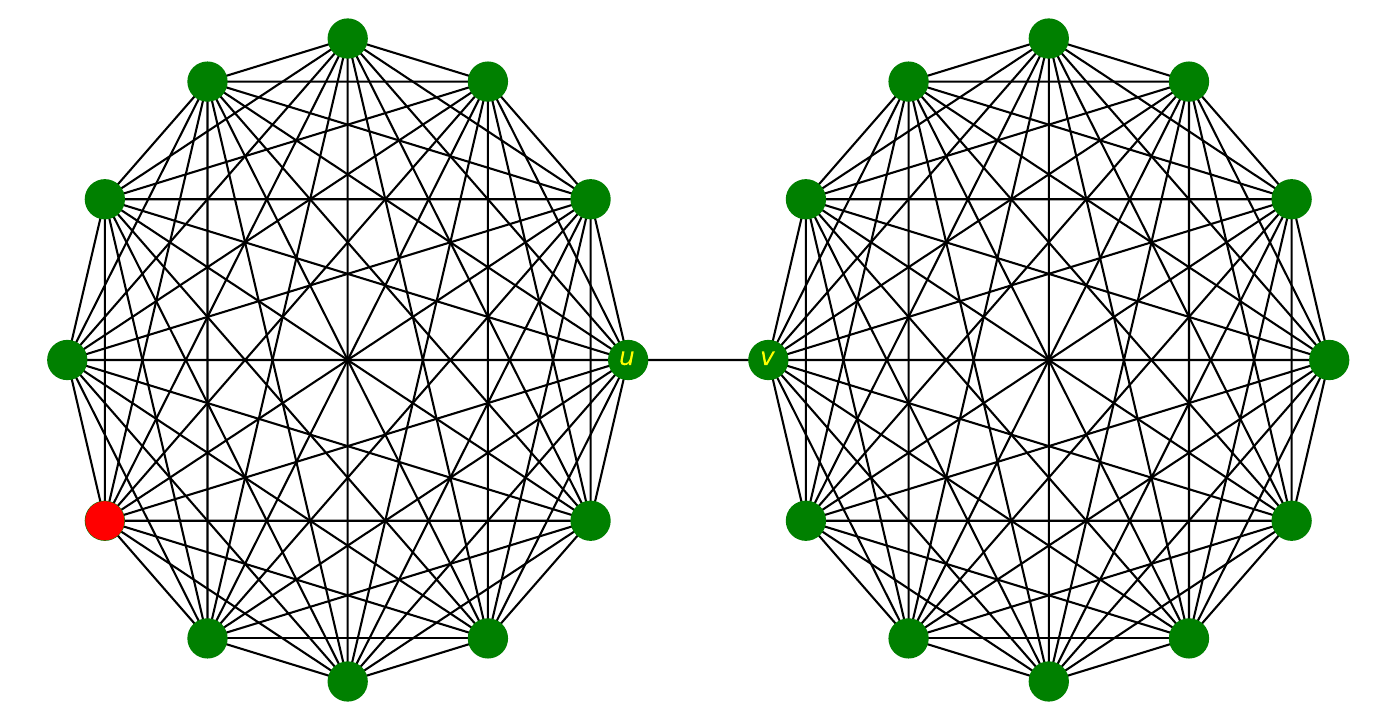}
    \caption{Two complete graphs connected by an edge. The red vertex is the ``marked'' vertex in the first graph.}
    \label{Fig:TwoCompleteGraphs}
\end{figure}

Consider two complete graphs, $G_1,G_2$, with $N$ vertices in each that are connected by a single edge.  Let $u \in G_1$ and $w \in G_2$ be the vertices connected by this single edge.  Assume further that a vertex $m \neq u$ is marked in graph $G_1$ (see Fig. \ref{Fig:TwoCompleteGraphs}).  We wish to examine the amplitudes of the ground state vector.  By symmetry, there are only four cases depending on the location of the vertices:  (1) $p \in G_1\backslash \{u\}$, (2) $p = u$, (3) $p = w$, and (4) $p \in G_2 \backslash \{w\}$. In what follows, $\lambda$ and $\phi$ represent the lowest eigenvalue and corresponding normalized eigenvector.
\begin{lemma}\label{Prop:TwoGraphs}
The following hold with the definitions from above:
\begin{itemize}
    \item[\emph{\textbf{(i)}}] $(2 - \lambda) \phi(p) = \phi(u)$,  $p \in G_1 \backslash \{u\}$
    \item[\emph{\textbf{(ii)}}] $ \left (N - \lambda - \frac{N-2}{N - \lambda} \right ) \phi(u) = \phi(w)$
    \item[\emph{\textbf{(iii)}}] $\frac{N - \lambda + \alpha  }{N-1}\phi(w) = \phi(v)$,  $v \in G_2 \backslash \{w\}$,  where $\alpha^{-1} = \left (N - \lambda - \frac{N-2}{N - \lambda} \right )$
    \item[\emph{\textbf{(iv)}}] $(1 - \lambda)\phi(q) = \phi(w)$, $q \in G_2 \backslash \{w\}$
\end{itemize}
\end{lemma}

\begin{proof}  Results (i) - (iv) all follow from considering expressions of the form
$$\sum_{y \sim x} \left ( \phi(x) - \phi(y) \right ) = \lambda \phi(x)$$
and invoking symmetry.
(i) follows from replacing $x$ with $p \in G_1 \backslash \{u\}$.  Expression (ii) is found by substituting $u$ for $x$ and applying (i).  (iii) is found similarly using $x = w$ and applying (ii). (iv) follows from setting $x = q \in G_2 \backslash \{w\}$ and applying (iii)
\end{proof}

Using the above lemma, we can show the relation between the relative amplitudes of the vertices from the four cases as well as their limiting behavior.

\begin{proposition}\label{Prop:Inequality}
With the definitions from above, $$\phi(p) < \phi(u) < \phi(w) < \phi(q)$$ where $p \in G_1 \backslash \{u\}$ and $q \in G_2 \backslash \{w\}$ are arbitrary.  As $N$ goes to infinity, $\phi(p)$ and $\phi(u)$ converge to 0 and $\phi(q)$ and $\phi(w)$ converge to $1/\sqrt{N}$.
\end{proposition}

\begin{proof}
  Using (i) from Lemma \ref{Prop:TwoGraphs} and the fact that $\lambda<1$, we have $\phi(p)<\phi(u)$.  Similarly, using (ii) from Lemma \ref{Prop:TwoGraphs}, $\phi(u)<\phi(w)$.  And finally, using (iv), $\phi(w)<\phi(q)$.
  
  From (ii) and (iii), as $N$ goes to infinity, $\phi(u)$ goes to 0 and $\phi(w)$ goes to $\phi(v)$.  Then $\phi(p)$ goes to 0 as well.  By symmetry $\phi(w)$ and $\phi(v)$ converge to $1/\sqrt{N}$.
\end{proof}

An immediate consequence of Proposition \ref{Prop:Inequality} is that, as $N$ goes to infinity, the spectrum associated with the final Hamiltonian associated with this graph converges to that of Grover's algorithm. Consequently, for large $N$, we may assume the optimal Grover schedule.

\end{document}